\documentclass[12pt,reqno]{amsart}
\usepackage{dsfont, amssymb,amsmath,amscd,latexsym, amsthm, amsxtra,amsfonts}
\usepackage[all]{xy}
\usepackage{graphicx}
\usepackage{float}
\usepackage{txfonts}
\usepackage[active]{srcltx}
\newtheorem{theorem}{Theorem}[section]
\newtheorem{lemma}[theorem]{Lemma}

\newtheorem{example}[theorem]{Example}

\newtheorem{Them}{Theorem}[section]

\begin{document}
\makeatletter
\def\@setauthors{%
\begingroup
\def\thanks{\protect\thanks@warning}%
\trivlist \centering\footnotesize \@topsep30\p@\relax
\advance\@topsep by -\baselineskip
\item\relax
\author@andify\authors
\def\\{\protect\linebreak}%
{\authors}%
\ifx\@empty\contribs \else ,\penalty-3 \space \@setcontribs
\@closetoccontribs \fi
\endtrivlist
\endgroup } \makeatother
 \baselineskip 17pt
\title[{{\small   Optimal  Control Strategies of Insurance Company   }}]
 {{\small Optimization of dividend and reinsurance strategies under ruin
  probability constraint}}
\author[{ Zongxia Liang and Jicheng Yao } ]
{ Zongxia Liang \\ Department of Mathematical Sciences, Tsinghua
University, Beijing 100084, China. Email:
zliang@math.tsinghua.edu.cn  \\Jicheng Yao \\ Department of
Mathematical Sciences, Tsinghua University, Beijing 100084, China.
Email: yaojicheng626@163.com} \maketitle
\begin{abstract}
This paper considers  nonlinear regular-singular stochastic optimal
control of large insurance company.  The company controls the
reinsurance rate and dividend payout process to maximize the
expected present value of the dividend pay-outs until the time of
bankruptcy. However, if the optimal dividend barrier is too low to
be acceptable, it will make the company result in bankruptcy soon.
Moreover, although risk and return should be highly correlated,
over-risking is not a good recipe for high return, the supervisors
of the company have to impose their preferred risk level and
additional charge on firm seeking services beyond or lower than the
preferred risk level. These indeed are nonlinear regular-singular
stochastic optimal problems under ruin probability constraints. This
paper aims at solving this kind of the optimal problems, that is,
deriving the optimal retention ratio,dividend payout level, optimal
return function and optimal control strategy of the insurance
company. As a by-product, the paper also sets a risk-based capital
standard to ensure the capital requirement of can cover the total
given risk, and the effect of the risk level on optimal retention
ratio, dividend payout level and optimal control strategy are also
presented. \vskip 10 pt
 \noindent {\bf MSC}(2000): Primary
91B30,91B70,93E20; Secondary   60H30, 60H10.
 \vskip 10pt
 \noindent
 {\bf Keywords:} Nonlinear regular-singular stochastic optimal
control; Ruin probability ;  Optimal retention ratio; Optimal
dividend payout; Optimal return function.
\end{abstract}
\vskip 10pt
 \setcounter{equation}{0}
\section{{\small {\bf Introduction}}}
 \vskip 15pt\noindent
In the present paper  we consider nonlinear  stochastic optimal
control  of insurance company.  The company controls the reinsurance
rate and dividend payout process to maximize the expected present
value of the dividend pay-outs until the time of bankruptcy. It is
well known that over-risking is not a good recipe for high return
although risk and return should be highly  correlated.
 In fact, to reduce the risk,  a risk-averse re-insurers may have
  their preferred risk
level and impose additional service charge on firms seeking services
beyond the target level, other re-insurers may demand additional
charges for those seeking services with risk level lower than its
preferred level as an aggressive move to gain market shares. This
indeed is nonlinear regular-singular stochastic optimal problem. The
objective of the company is to find a strategy, consisting of
optimal retention ratio and dividend payment scheme, which maximizes
the expected total discounted dividend pay-outs until the time of
bankruptcy. This is a mixed regular-singular control problem on
diffusion model which  has been a renewed interest recently,We refer
readers to He, Liang and et al.
\cite{ime01,ime02,ime03}(2008,2009)and references therein,
H{\o}jgaard and Taksar\cite{s3,s4,s1012}(1999, 1998, 2001),Asmussen
et all\cite{s7,s8}(1997,2000), Taksar\cite{s16}(2000), Guo Xin, Liu
Jun and Zhou Xunyu\cite{s1}(2004), Harrison and
Taksar\cite{s10101}(1983), Paulsen and Gjessing \cite{s1015}(1997),
and Radner and Sheep\cite{s1018}(1996),and other authors' works.
Recent surveys can  be found  in  Avanzi \cite{AV}. \vskip
15pt\noindent
 However, we notice that the optimal dividend
barrier in  the nonlinear regular-singular stochastic optimal
problem  may be so low that it would make the company result in
bankruptcy soon( see theorem \ref{theorem41}), the company may
reject this  optimal control strategy and may be prohibited to pay
dividend at such a low barrier because the insurance company is a
business affected with a public interest, and insureds and
policy-holders should be protected against insurer insolvencies (see
Williams and Heins\cite{law}(1985), Riegel and
Miller\cite{law2}(1963), and Welson and Taylor\cite{law3}(1959)).
The strategy,   making the company go bankrupt before termination of
contract between insurer and policy holders or the strategy of low
solvency(1 minus ruin probability(see \cite{law1})), is not the best
way and should be prohibited even though it can win the highest
profit. So the supervisor of the company will impose some
constraints on its ruin probability and find the best equilibrium
strategy between making profit and improving security. These are
turned out to be nonlinear regular-singular stochastic optimal
problems under low ruin probability constraint. This paper aims at
solving these kinds of stochastic optimal problems \vskip
15pt\noindent
 Unfortunately, there are very few results
concerning on these kinds of optimal control problem with lower
 ruin probability and higher security. He, Hou and
Liang\cite{ime02}(2008) investigated the optimal control problem for
linear Brownian model, Paulsen\cite{new23}(2003) and Taksar and
Markussen \cite{s6}(2004) studied also similar optimal controls
linear diffusion model via properties of return function. Since the
model treated in the present paper is very complicate and different
from He, Hou and Liang\cite{ime02}(2008) and
Paulsen\cite{new23}(2003), our results can not be directly deduced
from the \cite{ime02,new23}.
  Therefore, to solve these the problems we need to use initiated idea from
    the \cite{ime02}(2008), stochastic analysis and PDE method  to
  establish a complete setting for  further discussing
  optimal control problem  of a large insurance company under
 lower  ruin probability constraint.
 This paper is the first complete presentation of the topic, and
 the approach here is  rather general, so we
  anticipate that it can deal with
other models. We aim at deriving  the optimal return function, the
optimal retention rate and dividend payout level. The main result of
this paper will be presented in section 3 below. As a by-product,
the paper theoretically sets a risk-based capital standard to ensure
the capital requirement of can cover the total given risk. Moreover,
based on our main result, we also discuss how the risk affect the
optimal reactions of the insurance company by the implicit types of
solutions and how the optimal retention ratio, dividend payout level
and risk-based capital standard are affected by risk faced by the
insurance company, and how the initial capital and the premium rate
impact on the company's profit.
 \vskip 15pt\noindent
 The paper is organized as follows:  In next section, we establish
nonlinear  stochastic control model  of a large insurance company
 with  ruin probability constraint.
In section 3 we present  main result of this paper and its economic
and financial interpretations, and discuss how the risk affect the
optimal reactions of the insurance company by the implicit types of
solutions and how the optimal retention ratio, dividend payout level
and risk-based capital standard are affected by risk faced by the
insurance company, and how the initial capital and the premium rate
impact on the company's profit. In section 4 we give analysis on
risk of  stochastic control model treated in the present paper to
explain why we study nonlinear regular-singular stochastic optimal
control of insurance company. In section 5 we give some numerical
samples to portray how the risk  impacts on optimal dividend payout
level and  risk-based capital  based on PDE (\ref{eq36}) below, and
how the premium rate, preferred reinsurance level and volatility
effect on the company's profit. The proofs of theorems and lemmas
which study properties of probability of bankruptcy and optimal
return function will be given in  section 6 and appendix.
 \setcounter{equation}{0}
\section{{\small {\bf Nonlinear Mathematical Model  }}}
 \vskip15pt\noindent To give a
mathematical formulation of the stochastic control problem
 treated in this paper, let $(\Omega, \mathcal {F}, \{ \mathcal
{F}_{t}\}_{t\geq 0}, \mathbb{P})$ denote a filtered probability
space. $\{\mathcal {W}_{t}, t\geq 0\}$ is a standard Brownian motion
on this probability space.  $\mathcal {F}_{t}$ represents the
information available at time $t$ and any decision is made based on
this information. For the  {\sl intuition} of our diffusion model we
start from the classical Cram\'{e}r-Lundberg model of a
reserve(risk) process to portray that if the insurance company
shares risk with the reinsurance and takes no dividend pay-out then
its reserve process can be approximated by the following diffusion
process
\begin{eqnarray}\label{eq201} dR_{t}=\mu_1 U(t)dt+\sigma
U(t)d\mathcal {W}_{t},
\end{eqnarray}
where  $U(t)$ denotes retention level.
 \vskip 10pt\noindent
  In the
classical Cram\'{e}r-Lundberg model claims arrive according to a
Poisson process $N_t$ with intensity $\lambda $ on $(\Omega,
\mathcal {F}, \{ \mathcal {F}_{t}\}_{t\geq 0}, \mathbb{P})$. The
size of each claim is $X_i$. Positive random variables $X_i$ are
i.i.d. and are independent of the Poisson process $N_t$ with finite
first and second moments given by $\mu$ and $ \sigma^2$
respectively. If there is no reinsurance, dividend pay-outs, the
reserve (risk) process of insurance company is described by
$$ r_t= r_0 +pt-\sum^{N_t}_{i=1} X_i,$$
where $p$ is the premium rate. If $\eta > 0 $ denotes the {\sl
safety loading}, the $p$ can be calculated via the expected value
principle as
$$ p=(1+\eta ) \lambda \mu. $$
In a case where the insurance company shares risk with the
reinsurance, the sizes of the claims held by the insurer become
$X^{(U)}_i $, where $U$ is a (fixed) retention level.  For
proportional reinsurance, $U$ denotes the fraction of the claim
covered by the insurance company . Consider the case of {\sl cheap
reinsurance} for which the reinsuring company uses the same safety
loading as the insurance company, the reserve process of the
insurance company is given by
$$ r^{(U,\eta )}_t=u + p^{(U,\eta )}t - \sum^{N_t}_{i=1} X^{(U)}_i, $$
where
\begin{eqnarray*}
p^{(U,\eta)}=(1+\eta )\lambda \mathbb{E}\{X^{(U)}_i\}.
\end{eqnarray*}
Then by center limit theorem it is well known that for large enough
$\lambda $
\begin{eqnarray*}\label{eq204}
r_{t}^{(U,\eta)}\stackrel{d}{\approx} BM(\mu Ut,\sigma^{2}U^2t).
\end{eqnarray*}
in $\mathcal{D}[0, \infty)   $ (the space of right continuous
functions with left limits endowed with the skorohod topology),
where $\mu=\eta\lambda E(X_i)$,  $\sigma =\sqrt{\lambda E(X_i^2)}$
and  $ BM ( \mu, \sigma^2)$ stands for Brownian motion with the
drift coefficient $\mu $ and diffusion coefficient $\sigma $ on
$(\Omega, \mathcal {F}, \{ \mathcal {F}_{t}\}_{t\geq 0},
\mathbb{P})$. So the passage to the limit works well in the presence
of a big portfolios, the  reserve (risk) process of the insurance
company can be described by (\ref{eq201}). We refer the reader for
this fact and for the specifies of the diffusion approximations to
Emanuel, Harrison  and  Taylor \cite{C1}(1975),
Grandell\cite{C2,C3,C4}(1977,1978,1990), Harri-son \cite{C5}(1985),
Iglehart\cite{C6}(1969), Schmidli\cite{C7}(1994).
 \vskip 10pt\noindent
 It is well known that over-risking is not a good recipe for high
 return although risk and return are highly
correlated. This leads to question how an optimal strategy would
change when the risk and return are not linearly dependent on each
other. Moreover, while a risk-averse re-insurers may have their
preferred risk level and impose additional service charge on firms
seeking services beyond the target level, other re-insurers may
demand additional charges for those seeking services with risk level
lower than its preferred level as an aggressive move to gain market
shares. These make the  reserve process of the company should be the
following
\begin{eqnarray}\label{eq202}
dR(t)=[\mu_1 U(t)-a(U(t)-p)^2]dt+\sigma U(t)d\mathcal {W}_{t}, \
R(0)=x,
\end{eqnarray}
where p is the preferred reinsurance level imposed by the re-insurer
and $a$ is the additional rate of charge for the deviation from the
preferred level which  ensures that larger deviation is penalized
heavily. If we let
 $\mu=u_1+2ap$, $\delta=ap^2$, then the (\ref{eq202}) becomes
\begin{eqnarray}\label{eq203}
dR_{t}=(\mu U(t)-aU^2(t)-\delta)dt+\sigma U(t)d\mathcal {W}_{t},\ \
R(0)=x,
\end{eqnarray}
\vskip 15pt\noindent
 A  strategy $\pi$ is  a pair of non-negative
c\`{a}dl\`{a}g $ \mathcal {F}_{t}$-adapted  processes $\{U_\pi
(t),L_t ^\pi\}$, where $U_\pi (t)\in [l, 1](l>0)$ corresponds to the
risk exposure at time $t$ and $L_t ^ \pi$ corresponds to the
cumulative amount of dividend pay-outs distributed up to time $t$. A
strategy $\pi =\{U_\pi (t),L_t ^\pi\} $ is called admissible if $
l\leq U_\pi (t)\leq 1 $ and $L_t ^\pi$ is a nonnegative,
non-decreasing, right-continuous function. When $\pi
 $ is applied, the resulting reserve process is denoted by $\{ R_t^\pi
 \}$. We assume that  the initial reserve $R^\pi_0$ is  a
 deterministic value  $x$. In view of (\ref{eq203})
 the dynamics for  $R_t^\pi$ is given by
\begin{eqnarray}\label{eq204}
dR^\pi_{t}=[\mu U_\pi(t)-a(U_\pi(t))^2-\delta]dt+\sigma
U_\pi(t)d\mathcal {W}_{t}-dL^\pi_{t}, \ R^\pi_{0}=x.
\end{eqnarray}
In this case, we assume the company needs to keep its reserve above
$0$.  The company is considered ruin as soon as the reserves fall
below $0$. We define the time of bankruptcy by
$\tau^\pi_x=\inf\{t\geq 0: R^\pi_{t}\leq 0\}$. Obviously,
$\tau^\pi_x $ is an $\mathcal {F}_{t}$ -stopping time. So the
management of the insurance company should maximize the
 expected present value of the dividend payout by control strategy $ \pi $.
 Guo, Liu and Zhou\cite{s1} proved that there exists
a dividend level $b_0$, control strategy $ \pi^*_{b_0} $ and the
time of bankruptcy $\tau_x^{b_0}$ maximizing the expected present
value of the dividend payout before bankruptcy,
\begin{eqnarray}\label{eq205}
J(x, \pi )&=&{\bf E}\large [(\int_{0}^{\tau}e^{-cs}dL^\pi_{s})],
\end{eqnarray}
\begin{eqnarray}\label{eq206}
 V(x,b_0)&=&\sup\limits_{\pi\in \Pi } J(x,\pi)=J(x,\pi^*_{b_0}),
\end{eqnarray}
where $c$ denotes the discount rate, $\Pi$ is the set of all
admissible strategies. If the optimal dividend  level $b_0$ is
unacceptably low, then it will result in the company go to
bankruptcy early ( see theorem \ref{theorem41} below). To take
security and solvency into consideration and set a risk-based
capital and dividend  standard to ensure the capital and dividend
requirement of can cover the total risk, we introduce our optimal
control problem of nonlinear stochastic model (\ref{eq204}) as
follows.\vskip 5pt\noindent Let $\Pi_b=\{\pi\in \Pi : \int _0 ^{
\infty }I_{\{s:R^\pi(s)<b\}}dL_s^{\pi}=0\} $ for $b\geq 0$ . Then it
is easy to see that $\Pi=\Pi_0$ and $b_1>b_2\Rightarrow
\Pi_{b_1}\subset \Pi_{b_2}$. For a given admissible strategy $\pi$
we define the optimal return function $V(x)$ by
\begin{eqnarray}\label{eq207}
J(x,\pi)&=&{\bf E}\big \{\int_0^{\tau^\pi_x } e^{-ct} dL_t^\pi\big\},\nonumber \\
 V(x,b )&=&\sup_{\pi \in \Pi_b}\{J(x,\pi)\},
\end{eqnarray}
\begin{eqnarray}\label{eq208}
 V(x)&=&\sup_{b \in \mathfrak{B}}\{  V(x,b )  \}
\end{eqnarray}
and the optimal strategy $\pi^* $ by
\begin{eqnarray}\label{eq209}
J(x,\pi^*)= V(x),
\end{eqnarray}
 where
  \begin{eqnarray*}
 \mathfrak{B}:=\big \{b\ :\ \mathbb{P}[\tau_{b}^{\pi_{b}} \leq T] \leq
\varepsilon \ , \  J(x, \pi_b)= V(x,b) \mbox{ and} \ \pi_b \in
\Pi_b\big\},
  \end{eqnarray*}
 $c>0$ is a discount rate, $\tau_b^{\pi_b}$ is
 the time of bankruptcy $\tau_x^{\pi_b} $  when
the initial reserve $x=b$ and the control strategy is $\pi_b$.
$1-\varepsilon$ is the standard of security and less than solvency
for given risk level $\varepsilon>0 $. \vskip 15pt\noindent {\sl The
main purpose of this paper is to derive the optimal return function
$V(x)$, the optimal retention rate $U^*(t)$ and dividend payout
level $b^*$ as well as a risk-based capital $x(\varepsilon,b^*)$ to
ensure the capital requirement of can cover the total risk}
$\varepsilon$.
 \setcounter{equation}{0}
\section{ \bf Main result }
\vskip 15pt\noindent
In this section we first present main result of
this paper, then give its economic and financial interpretations .
\begin{theorem}\label{theorem31}
Let  level of risk $ \varepsilon \in (0, 1)$ and time horizon $T$ be
given. \vskip 10pt\noindent (i) If $
\mathbf{P}[\tau_{b_0}^{\pi^*_{b_0}}\leq T]\leq \varepsilon$,then the
optimal return function $V(x)$ is $f(b_0, x)$ defined by
(\ref{eq16}) below, and $V(x)=f(b_0, x)=J(x,\pi_{b_o}^\ast) $. The
optimal  strategy $\pi_{b_o}^\ast $ is
$\{U^\ast(R^{\pi_{b_o}^\ast}_t),L^{\pi_{b_o}^\ast}_t\}$, where
$\{R^{\pi_{b_o}^\ast}_t, L^{\pi_{b_o}^\ast}_t  \} $ is uniquely
determined  by the following stochastic differential
 equation
\begin{eqnarray}\label{eq301}
\left\{
\begin{array}{l l l}
dR_t^{\pi_{b_o}^\ast}=(\mu
U^*_{b_0}(R^{\pi_{b_o}^\ast}_t)-a{U^*_{b_0}}^2(R^{\pi_{b_o}^\ast}_t)-\delta
)dt+\sigma U^*_{b_0}(R^{\pi_{b_o}^\ast}_t)d
{W}_{t}-dL_t^{\pi_{b_o}^\ast},\\
R_0^{\pi_{b_o}^\ast}=x,\\
0\leq R^{\pi_{b_o}^\ast}_t\leq  b_0,\\
\int^{\infty}_0 I_{\{t: R^{\pi_{b_o}^\ast}_t
<b_0\}}(t)dL_t^{\pi_{b_o}^\ast}=0.
\end{array}\right.
\end{eqnarray}
The solvency of the company is bigger than $1-\varepsilon$. \vskip
10pt\noindent (ii)If $ \mathbf{P}[\tau_{b_0}^{\pi^*_{b_0}}\leq
T]>\varepsilon $, there is a unique optimal dividend $b^\ast(\geq
b_0)$ satisfying $\mathbb{P}[\tau _{b^\ast}^{\pi_{b^*}^\ast}\leq T]=
\varepsilon $. The optimal return  function $V(x)$ is $g(x,b^*)$
defined by (\ref{eq22}), that is,
\begin{eqnarray}\label{eq302}
V(x)= g(x,b^*) =\sup_{b\in \mathfrak{B}}\{V(x,b)\},
\end{eqnarray}
and
\begin{eqnarray}\label{eq303}
b^* \in \mathfrak{B}:=\big \{b: \mathbb{P}[\tau_{b}^{\pi^*_{b}} \leq
T] \leq \varepsilon, \  J(x, \pi^*_b)= V(x,b) \mbox{ and} \ \pi^*_b
\in \Pi_b\ \big\}.
\end{eqnarray}
The optimal strategy $\pi_{b^*}^\ast$ is $
 \{U^\ast_{b^*}(R^{\pi_{b^*}^\ast}_t),L^{\pi_{b^*}^\ast}_t\}$, where
$\{R^{\pi_{b^*}^\ast}_t, L^{\pi_{b^*}^\ast}_t  \} $ is uniquely
determined  by the following stochastic differential
 equation
\begin{eqnarray}\label{eq304}
\left\{
\begin{array}{l l l}
dR_t^{\pi_{b^*}^\ast}=(\mu
U^*_{b^*}(R^{\pi_{b^*}^\ast}_t)-a{U^*_{b^*}}^2(R^{\pi_{b^*}^\ast}_t)-\delta
)dt+\sigma U^*_{b^*}(R^{\pi_{b^*}^\ast}_t)d
{W}_{t}-dL_t^{\pi_{b^*}^\ast},\\
R_0^{\pi_{b^*}^\ast}=x,\\
0\leq R^{\pi_{b^*}^\ast}_t\leq  b^*,\\
\int^{\infty}_0 I_{\{t: R^{\pi_{b^*}^\ast}_t
<b^*\}}(t)dL_t^{\pi_{b^*}^\ast}=0.
\end{array}\right.
\end{eqnarray}
The solvency of the company is $1-\varepsilon$. \vskip 10pt\noindent
(3) Moreover,
\begin{eqnarray}\label{eq305}
\frac{g(x,b^*)}{g(x,b_0)}\leq 1.
 \end{eqnarray}
 Where $U^*_b(x)$ is defined by (\ref{eq28}) and (\ref{eq21}) below.
 \end{theorem}
\vskip 0.5cm\noindent { \bf  Economic and financial explanation of
theorem \ref{theorem31} is as follows.}\vskip 10pt\noindent
 {\sl (1) For a given level of
risk  and time horizon, if ruin probability  is less than the level
of risk, the optimal
  control problem of
(\ref{eq207}) and (\ref{eq208}) is the traditional  (\ref{eq205})
and (\ref{eq206}), the company has higher solvency, so it will have
good reputation. The solvency constraints here do not work. This is
a trivial case. \vskip 10pt\noindent
 (2) If ruin probability is
large than the level of risk $ \varepsilon$, the traditional optimal
strategy will not meet the standard of security and solvency, the
company needs to find a sub-optimal strategy $\pi_{b^*}^\ast $ to
improve its solvency. The sub-optimal reserve process $
R^{\pi_{b^*}^\ast}_t $ is a diffusion process reflected at $b^*$,
the process $L^{\pi_{b^*}^\ast}_t  $ is the process which ensures
the reflection. The sub-optimal action is to pay out everything in
excess of $b^*$ as dividend and pay no dividend when the reserve is
below $b^*$, and $ U_{b^*}^*(x)$ is the sub-optimal feedback control
function. The solvency is $1-\varepsilon$.  \vskip 10pt\noindent (3)
On the one hand, the inequality (\ref{eq305}) states that
$\pi_{b^*}^\ast $ will reduce the company's profit, on the other
hand, in view of (\ref{eq305}) and $\mathbb{P}[\tau
_{b^\ast}^{\pi_{b^*}^\ast}\leq T]= \varepsilon $ as well as lemma
\ref{le04} below, the cost of improving solvency is minimal
$g(x,b_0)-g(x, b^*)$. Therefore the strategy $\pi_{b^*}^\ast $ is
the best equilibrium action between making profit and improving
solvency. \vskip 10pt\noindent (4) The risk-based capital
$x(\varepsilon,b^*)$ to ensure the capital requirement of can cover
the total risk $\varepsilon$ can be determined by numerical solution
of $1-\phi^{b^*}(x,b^*)=\varepsilon$ based on (\ref{eq36}). We see
from the figure \ref{x-epsilon} that risk-based capital
$x(\varepsilon,b^*)$ decreases with risk $\varepsilon  $, i.e.,
$x(\varepsilon,b^*)$ increases with solvency , so does risk-based
dividend level $b^*(\varepsilon  )$ (see the figure
\ref{b-epsilon}). \vskip 10pt\noindent (5) We also see from the
figures \ref{difmu} and \ref{difsigma} below that the premium rate
will increase the company's profit, higher risk will get higher
return. \vskip 10pt\noindent (6) We also see from the figure
\ref{difp} below shows that the value function $g(x,p)$ increases
with $(x,p)$, i.e., the initial capital and the premium rate will
increases the company's profit.} \vskip 10pt\noindent
 \setcounter{equation}{0}
\section{\bf Analysis of  risk on model (\ref{eq204})}
 \vskip 10pt\noindent
 The first result of this
section is the following, which states that the company has to find
optimal strategy to improve its solvency.
\begin{Them}\label{theorem41}
Let
 $\{R^{\pi_{b_0}^\ast}_t, L^{\pi_{b_0}^\ast}_t  \} $ be defined  by the
following SDE( see Lions and Sznitman \cite{Lions}(1984))
\begin{eqnarray}\label{eq401-theorem41}
\left\{
\begin{array}{l l l}
dR_t^{\pi_{b_o}^\ast}=(\mu
U^*_{b_0}(R^{\pi_{b_o}^\ast}_t)-a{U^*_{b_0}}^2(R^{\pi_{b_o}^\ast}_t)-\delta
)dt+\sigma U^*_{b_0}(R^{\pi_{b_o}^\ast}_t)d
{W}_{t}-dL_t^{\pi_{b_o}^\ast},\\
R_0^{\pi_{b_o}^\ast}=x,\\
0\leq R^{\pi_{b_o}^\ast}_t\leq  b_0,\\
\int^{\infty}_0 I_{\{t: R^{\pi_{b_o}^\ast}_t
<b_0\}}(t)dL_t^{\pi_{b_o}^\ast}=0.
\end{array}\right.
\end{eqnarray}
Then for any $x\in (0, b_0]$ we have
\begin{eqnarray}\label{eq402-theorem41} {\bf P}(\tau_{x}^{b_0}\leq T)\geq
 \varepsilon_0(b_0, \sigma^2, \mu, p, l,a)
 \equiv\frac{4[1-\Phi(\frac{b_0}{l\sigma\sqrt{T}})]^2}{\exp\{
\frac{(\mu-a-\delta)^2 T}{\sigma^2}\}}>0,
\end{eqnarray}
where $\Phi(\cdot)$ is the standard normal distribution function.
\end{Them}
\noindent {\bf  The economic interpretation of  theorem
 \ref{theorem41} is the following.}\vskip 10pt \noindent
 {\sl   (1)\ The lower  boundary
 $\varepsilon_0(b_0, \sigma^2, \mu, p, l,a)$ of ruin probability
  for the company
is an increasing function of $(\sigma^2,l)$, thus higher volatility
$\sigma^2$ and  fraction of the claim covered by the company will
make the company have larger risk. \vskip 10pt \noindent(2) \ The
lower boundary $\varepsilon_0(b_0, \sigma^2, \mu, p, l,a)$ of ruin
probability for the company is a decreasing function of $ ( b_0,
\mu, p,a )$, so early making dividend will increasing the company's
risk. The premium rate, preferred reinsurance level and additional
rate of charge for the deviation from the preferred level
 will decrease the company's risk.} \vskip 10pt\noindent
\begin{proof}
Let $\{ R^{(1)}_t\}$  be a stochastic process satisfying
\begin{eqnarray}\label{eq403-theorem41}
\left\{
\begin{array}{l l l}
dR^{(1)}_t=(\mu
U^*_{b_0}(R^{(1)}_t)-a{U^*_{b_0}}^2(R^{(1)}_t)-\delta )dt+\sigma
U^*_{b_0}(R^{(1)}_t)d
{W}_{t},\\
R^{(1)}_0=b_0\\
\end{array}\right.
\end{eqnarray}
where $ U^*_{b_0}(\cdot) $ is defined by (\ref{eq21}). Define a
measure ${\bf Q}$ on $\mathcal{ F}_T  $ by
$$ d{\bf Q}(\omega)=M_1(T)d P(\omega )$$
where
\begin{eqnarray*}\label{eq404-theorem41}
M_1(t)&\equiv&\exp\big\{-\int_{0}^{t}\frac{(\mu
U^*_{b_0}(R_t^{(1)})-a[U^*_{b_0}(R_t^{(1)})]^2-\delta)}{\sigma
U^*_{b_0}(R_t^{(1)})}dW_{s}\\
&-& \frac{1}{2}\int_{0}^{t}\frac{(\mu
U^*_{b_0}(R_t^{(1)})-a[U^*_{b_0}(R_t^{(1)})]^2-\delta)^{2}}{[\sigma
U^*_{b_0}(R_t^{(1)})]^{2}}ds\big \}.
\end{eqnarray*}
Since $\{ M_1(t)\}$ is a martingale w.r.t.$\mathcal {F}_t$,  ${\bf
E} \big [ M_1(T) \big ] =1$. Using Girsanov theorem,  ${\bf Q}$ is a
probability measure on $\mathcal {F}_T$ and the process
$\{R^{(1)}_t\}$ satisfies the following SDE
\begin{eqnarray}\label{eq405-theorem41}
dR_t^{(1)}=U^*_{b_0}(R_t^{(1)})\sigma d\tilde{W}_t,R_0^{(1)}=b_0
\end{eqnarray}
where $\tilde{W}_t$ is a Brownian motion  on $(\Omega, \mathcal {F},
\{ \mathcal {F}_{t}\}_{t\geq 0}, {\bf Q})$. \vskip 5pt\noindent
Define a time changes $\rho(t)$ by
\begin{eqnarray}\label{eq406-theorem41}
\dot{\rho}(t)=\frac{1}{{U^*_{b_0}}^2(R_t^{(1)})\sigma^2},
\end{eqnarray}
and $\hat{R}_t^{(1)}$ by $R_{\rho(t)}^{(1)}$. Then $\rho(t)$ is a
strictly increasing function  and
\begin{eqnarray*}
\hat{R}_t^{(1)}=b+\hat{W}_t
\end{eqnarray*}
where $\hat{W}_t$ is also a standard Brownian motion on $(\Omega,
\mathcal {F}, \{ \mathcal {F}_{t}\}_{t\geq 0}, {\bf Q})$. Noticing
that ${U^*_{b_0}}^2(R_t^{(1)})\geq l^2>0$, where  $l$ is a positive
low boundary of optimal retention ratio $U^*_{b_0}\in [l, 1]$, we
have
\begin{eqnarray}\label{eq407-theorem41}
\dot{\rho}(t)\leq \frac{1}{l^2\sigma^2}.
\end{eqnarray}
Moreover, $\rho(t)\leq \frac{1}{l^2\sigma^2} t$ and  $\rho^{-1}(t)
\geq l^2\sigma^2 t$. So
\begin{eqnarray*}
{\bf{Q}}[\tau^{(1)}\leq T]&=&{\bf{Q}}[\inf\{t:
\hat{R}_{\rho^{-1}(t)}^{(1)}\leq 0\}\leq
T]\\&=&{\bf{Q}}[\inf\{\rho(t): b_0+\hat{W}_t\leq 0\}\leq T]\\
&=&{\bf{Q}}[\inf \{t: \hat{W}_t\leq - b_0 \}\leq
\rho^{-1}(T)]\\&\geq& {\bf{Q}}[\inf\{t: \hat{W}_t\leq -b_0\}\leq
l^2\sigma^2 T]\\&=&2[1-\Phi(\frac{b_0}{l\sigma\sqrt{T}})]>0,
\end{eqnarray*}
where $\tau^{(1)}=\inf\{t>0; R_t^{(1)}\leq 0\}$ is a stopping time.
Using comparison theorem for one-dimensional It\^{o} process, we
have ${\bf{P}}[R^{\pi_{b_o}^\ast}_t \leq R_t^{(1)}]=1$. By
${\bf{E}}^{\bf{P}}[M_1(T)^{2}]\leq \exp\{ \frac{(\mu-a-\delta)^2
T}{\sigma^2}\}$ and  H\"{o}lder inequalities we have
\begin{eqnarray*}
{\bf{P}}[\tau_{x}^{b_0}\leq T]&\geq&{\bf{P}}[\tau_{b_0}^{b_0}\leq
T]\\&\geq&{\bf{P}}[\tau^{(1)}\leq T]
\\&\geq&{\bf{Q}}[\tau^{(1)}\leq
T]^2/{\bf{E}}^{{\bf{P}}}[M_1(T)^{2}]\\
&\geq& \frac{4[1-\Phi(\frac{b_0}{l\sigma\sqrt{T}})]^2}{\exp\{
\frac{(\mu-a-\delta)^2 T}{\sigma^2}\}}>0.
\end{eqnarray*}
\end{proof}
\vskip 10pt \noindent
 The second result of this section is the following. It sates that
 the  restrained
 set $\mathfrak{B} $ above is non-empty for any $ \varepsilon >0$. So the
 (\ref{eq207}),(\ref{eq208}) and (\ref{eq209}) are well defined.
\begin{Them}\label{theorem42}
Let $( R^{\pi_{b}^\ast}_t, L_t^{\pi_{b}^\ast}    )$ be defined by
\begin{eqnarray}\label{eq401-theorem42}
\left\{
\begin{array}{l l l}
dR_t^{\pi_{b}^\ast}=(\mu
U^*_{b}(R^{\pi_{b}^\ast}_t)-a{U^*_{b}}^2(R^{\pi_{b}^\ast}_t)-\delta
)dt+\sigma U^*_{b}(R^{\pi_{b}^\ast}_t)d
{W}_{t}-dL_t^{\pi_{b}^\ast},\\
R_0^{\pi_{b}^\ast}=b,\\
0\leq R^{\pi_{b}^\ast}_t\leq  b,\\
\int^{\infty}_0 I_{\{t: R^{\pi_{b}^\ast}_t
<b\}}(t)dL_t^{\pi_{b}^\ast}=0,
\end{array}\right.
\end{eqnarray}
and $ \tau^b_b:=\tau_{b}^{\pi_b^*}=\inf\{t\geq 0:
R^{\pi_{b}^\ast}_t<0 \}$. Then
\begin{eqnarray}\label{eq402-theorem42}
\lim_{b\rightarrow  \infty}{\bf P}[ \tau_{b}^{b}\leq T]=0.
\end{eqnarray}
\end{Them}
\begin{proof}
Let $x_2$ be defined as in (\ref{eqx2}). For  $b> x_2 $, by
comparison theorem for SDE, we have
\begin{eqnarray*}
{\bf P}\{\tau_{b}^{b}\leq T\}\leq {\bf P}\{\tau_{(b+x_2)/2}^{b}\leq
T\}.
\end{eqnarray*}
It is easy to see that
\begin{eqnarray*}
{\bf P}\{\tau_{(b+x_2)/2}^{\pi^*_b}\leq T\} &\leq&{\bf
P}\{R_t^{(2)}=x_2 \ \mbox{or}\ R_t^{(2)}
=b\ \mbox{for some  $t\geq 0$ }\}\nonumber\\
&\leq& {\bf P}\{\sup_{0\leq t \leq T} R_t^{(2)}\geq b\}+{\bf
P}\{\inf_{0\leq t \leq T} R_t^{(2)}\leq x_2\},
\end{eqnarray*}
where $ \{R_t^{(2)}\}$ is the unique solution of the following SDE
\begin{eqnarray}\label{eq403-theorem42}
\left\{
\begin{array}{l l l}
dR^{(2)}_{t}=(\mu
U^*_{b}(R^{(2)}_{t})-a{U^*_{b}}^2(R^{(2)}_{t})-\delta )dt+\sigma
U^*_{b}(R^{(2)}_{t})d{W}_{t},\\
 R^{(2)}_{0}=(b+x_2)/2.
\end{array}\right.
\end{eqnarray}
Define a measure  $\mathbb{Q}_1$ on  $\mathcal{F}_T $  by
\begin{eqnarray*}\label{eq404-theorem42}
d{\bf P}(\omega)=\widetilde{M}_T(\omega)d \mathbb{Q}_1(\omega),
\end{eqnarray*}
where
\begin{eqnarray*}
\widetilde{M}_{t}&=&\exp\big\{\int_{0}^{t}\frac{(\mu
U^*_{b}(R^{(2)}_{t})-a{U^*_{b}}^2(R^{(2)}_{t})-\delta )}{\sigma
U_{b}^*(R_t^{(2)})}dW_{s} \\
&+&\frac{1}{2}\int_{0}^{t}\frac{(\mu
U^*_{b}(R^{(2)}_{t})-a{U^*_{b}}^2(R^{(2)}_{t})-\delta )^{2}}{[\sigma
U_{b}^*(R_t^{(2)})]^{2}}ds\big \}
\end{eqnarray*}
is a martingale. Then $\mathbb{Q}_1$ is a probability measure on
$\mathcal {F}_{T}$. By Girsanov theorem
\begin{eqnarray*}
\hat{\mathcal{W}_{t}}:=\int_{0}^{t}\frac{(\mu
U^*_{b}(R^{(2)}_{t})-a{U^*_{b}}^2(R^{(2)}_{t})-\delta )}{\sigma
U_{b}^*(R_s^{(2)})}ds+\mathcal{W}_{t} , t\leq T
\end{eqnarray*}
is a standard Brownian motion on $(\Omega, \mathcal {F}, \{ \mathcal
{F}_{t}\}_{t\geq 0}, \mathbb{Q}_1)$. So  the (\ref{eq403-theorem42})
becomes
\begin{eqnarray*}
dR_t^{(2)}=\sigma U^*_{b}(R^{(2)}_{t})\hat{\mathcal{W}_{t}}, \
R^{(2)}_{0}=(b+x_2)/2.\ \mbox{a.e.}, \  \mathbb{Q}_1
\end{eqnarray*}
Firstly, we now estimate   $ {\bf P}\{\sup_{0\leq t \leq T}
R_t^{(2)}\geq b\}$. By SDE (\ref{eq403-theorem42}), H\"{o}lder's
inequalities,Chebyshev inequalities and B-D-G inequalities for
martingales (see  Ikeda and Watanabe \cite{s107}(1981))
\begin{eqnarray}\label{eq405-theorem42}
 {\bf P}\{\sup_{0\leq t \leq
T} R_t^{(2)}\geq b\}&\leq & [{\bf E}^{\mathbb{Q}_1}\{
\widetilde{M}^2_T \}]^{\frac{1}{2}} \mathbb{Q}_1\{ \sup_{0\leq t
\leq T} R_t^{(2)}\geq b    \}^{\frac{1}{2}}£¬
\end{eqnarray}
and
\begin{eqnarray}  \label{eq406-theorem42}
\mathbb{Q}_1\{\sup_{0\leq t\leq T}R_t^{(2)}\geq b \} &\leq
&{\mathbb{Q}_1}\{\sup_{0\leq t\leq T}|\int_0^t \sigma U^*_{b}
(R^{(2)}_{s})d\hat{\mathcal{W}_{s}} |\geq \frac{b-x_2}{2}\}\nonumber\\
&\leq &\frac{ 4{\bf E}^{\mathbb{Q}_1} \{ \sup_{0\leq t \leq
T}|\int_0^t \sigma U^*_{b}(R^{(2)}_{s})
d\hat{\mathcal{W}_{s}}|    \}^2}{(b-x_2)^2 }\nonumber\\
&\leq  & \frac{ 16{\bf E}^{\mathbb{Q}_1} \{ \int_0^T ( \sigma
U^*_{b} (R^{(2)}_{s}))^2
 ds|    \}}{(b-x_2)^2 }\nonumber\\
 &\leq & \frac{16T\sigma^2 \widetilde{B}^2}{(b-x_2)^2},
\end{eqnarray}
where ${\bf E}^{\mathbb{Q}_1}$ denotes the expectation w.r.t. $\mathbb{Q}_1$ .\\
Next we estimate   $P\{\inf_{0\leq t \leq T} R_t^{(1)}\leq x_2\} $.
Since  $U^*_{b} (x)=1$  for $x\geq x_2$,
\begin{eqnarray}\label{eq407-theorem42}
 {\bf P}\{\inf_{0\leq t \leq T}
R_t^{(2)}\leq x_2\}&=&1-\mathbf{P}\{\inf_{0\leq t \leq
T} R_t^{(2)}> x_2\}\nonumber \\
&=&1- {\bf P}\{\inf_{0\leq t\leq T}\{\mu t +\sigma
W_t\}>-\frac{b-x_2}{2}
\}\nonumber\\
&&\rightarrow 1-1=0 \ \mbox{as } \ b\rightarrow \infty.
\end{eqnarray}
Finally, since $l \leq U^*_{b} \leq 1$,
\begin{eqnarray}\label{eq407-theorem42}
{\bf E}^{\mathbb{Q}_1}\{ \widetilde{M}^2_T \}\leq C(T)< \infty.
\end{eqnarray}
So the equation (\ref{eq402-theorem42}) follows from
(\ref{eq405-theorem42})-(\ref{eq407-theorem42}).
 \end{proof}
 \setcounter{equation}{0}
\section{\bf Numerical examples}
\vskip 5pt \noindent In this section we consider some numerical
samples to demonstrate  how the risk $\varepsilon$ impacts on
optimal dividend payout level $b^*$ and risk-based capital $x$ based
on PDE (\ref{eq36}) below, and how the premium rate, preferred
reinsurance level and volatility effect on the company's profit.
\begin{example}
Let $\mu=2$, $\sigma^2=50$, $ l=0.5$,$ a=0.1$, $\delta=0.01$,
$c=0.05$, $T=500$ and solve $b( \varepsilon )$ by
$1-\phi(T,b)=\varepsilon$, we get the figure \ref{b-epsilon} below.
It shows that the risk $ \varepsilon$ greatly impacts on dividend
payout level $b$. The dividend payout level $b$ decreases with the
risk $ \varepsilon$, so the risk $ \varepsilon$ increases the
company's profit.
\end{example}
\begin{figure}[H]
\includegraphics[width=0.7\textwidth]{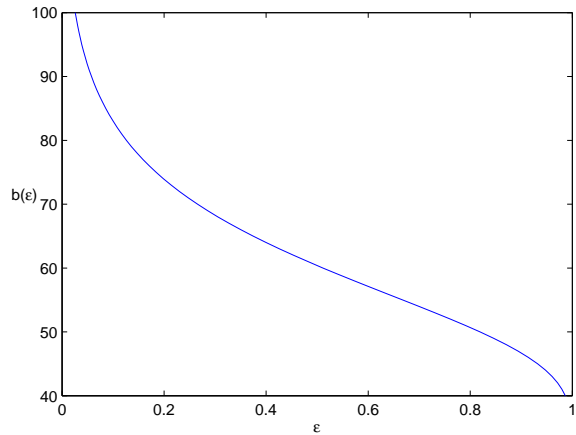}
\caption{ Dividend payout level $b$ as a function of $ \varepsilon$
(Parameters: $\mu=2$, $\sigma^2=50$, $ l=0.5$,$ a=0.1$,
$\delta=0.01$, $c=0.05$, $T=500$)}\label{b-epsilon}
\end{figure}
\begin{example}
Let $ b=100, \sigma^2=50, l=0.5, a=0.1, \delta=0.01, c=0.05, T=500 $
the figure \ref{difmu} below shows that the value function
$g(x,\mu)$ increases with $(x,\mu)$, so does the company's profit.
\end{example}
\begin{figure}[H]
\includegraphics[width=0.7\textwidth]{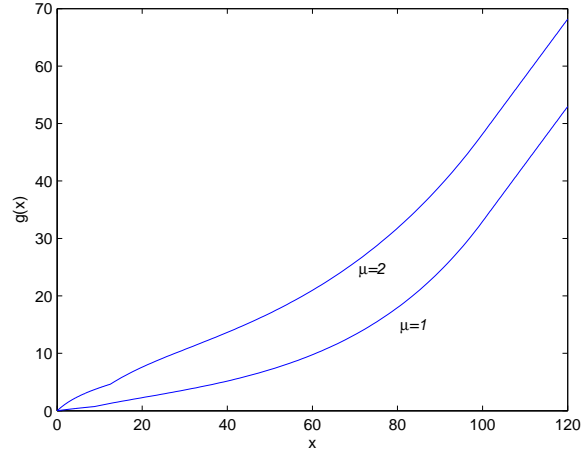}
\caption{ Value function $g(x,\mu)$ as a function of $(x,\mu)$
(Parameters: $ b=100, \sigma^2=50, l=0.5, a=0.1, \delta=0.01,
c=0.05, T=500 $ )}\label{difmu}
\end{figure}
\begin{example}
Let $\mu=2, b=100, \sigma^2=50, l=0.5, a=0.5,c=0.05, T=500 $. The
figure \ref{difp} below shows that the value function $g(x,p)$
increases with $(x,p)$, so does the company's profit.
\end{example}
\begin{figure}[H]
\includegraphics[width=0.7\textwidth]{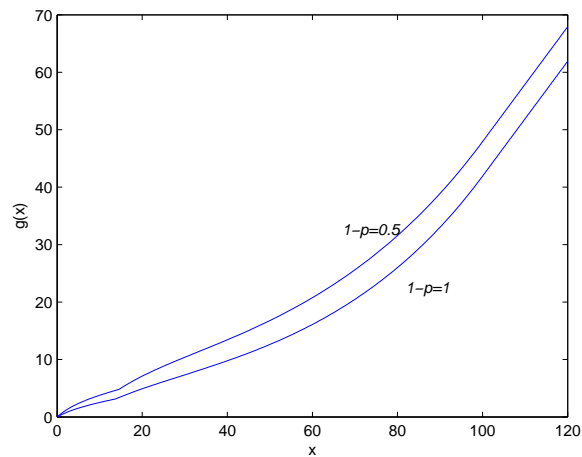}
\caption{ Value function $g(x,p)$ as a function of $(x,p)$
(Parameters: $\mu=2, b=100, \sigma^2=50, l=0.5, a=0.5, c=0.05, T=500
$ )}\label{difp}
\end{figure}
\begin{example}
Let $\mu=2, b=100, l=0.5, a=0.1, \delta=0.01, c=0.05, T=500$. The
figure \ref{difsigma} below shows that the value function
$g(x,\sigma^2)$ increases with $(x, \sigma^2)$, so does the
company's profit.
\end{example}
\begin{figure}[H]
\includegraphics[width=0.7\textwidth]{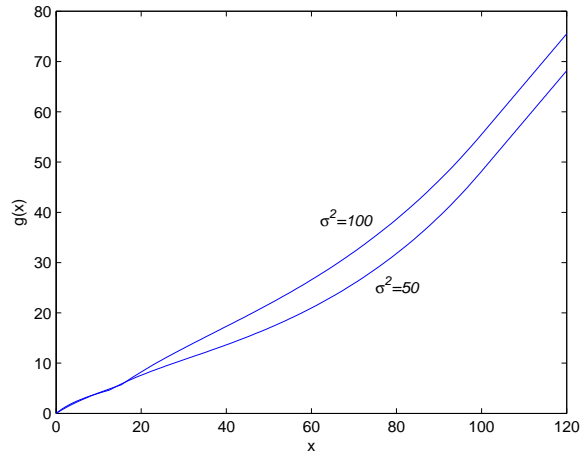}
\caption{ Value function $g(x,\sigma^2)$ as a function of
$(x,\sigma^2)$ (Parameters: $\mu=2, b=100, l=0.5, a=0.1,
\delta=0.01, c=0.05, T=500$)}\label{difsigma}
\end{figure}
\begin{example}
Let $\mu=2, \sigma^2=50, l=0.5, a=0.1, \delta=0.01, c=0.05, T=500 $.
The figure \ref{x-epsilon} below shows that the initial capital $x(
\epsilon )$ decreases with $ \epsilon$.
\end{example}
\begin{figure}[H]
\includegraphics[width=0.7\textwidth]{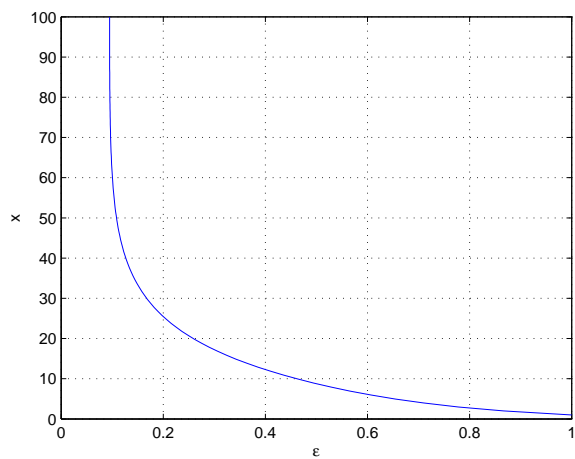}
\caption{ Initial capital $x( \epsilon )$ as a function of
$\epsilon$ (Parameters: $\mu=2, \sigma^2=50, l=0.5, a=0.1,
\delta=0.01, c=0.05, T=500 $)}\label{x-epsilon}
\end{figure}
\setcounter{equation}{0}
\section{\bf Properties on $V(x,b)$ and ruin  probability }
\vskip 10pt\noindent In this section, to prove Theorem 3.1, we list
some lemmas on properties of   $V(x,b)$ and ruin probability which
will be used late. The rigorous proofs of these lemmas will be given
in the appendix below. Throughout this paper we assume that $
\mu/2a>1$ and $ 0<l\leq U(t)\leq 1$.
\begin{lemma}\label{le01} There exists
$b_0>0 $ such that if  $f(x)\in C^{2}$ satisfies the following HJB
equations and the boundary conditions,
\begin{eqnarray}\label{eq16}
&&\max\limits_{U\in[l,1]}[\frac{1}{2}\sigma^{2}U^{2}f^{''}(x)+(\mu
U-aU^2-\delta)f^{'}(x)-cf(x)]=0,\\
&& \mbox{ for}\ \  0\leq x\leq
b_{0},\nonumber\\
&& f^{'}(x)=1, \ \mbox{for}\   x\geq b_{0}, \nonumber\\
&& f^{''}(x)=0,\ \mbox {for}  \ x\geq b_{0},\nonumber\\
&& f(0)=0,\nonumber
\end{eqnarray}
then we have the following,
\begin{eqnarray*}
&&\max\mathcal {L}f(x)\leq 0, f^{'}(x)\geq 1,\ \mbox{ for} \   x
\geq
0,\\
&& f(0)=0,\nonumber
\end{eqnarray*}
where
 $\mathcal {L}= \frac{1}{2}\sigma^{2}U^{2}\frac{d^{2}}{d
x^{2}}+(\mu U-aU^2-\delta)\frac{d}{d x}-c$.
\end{lemma}
\begin{lemma}\label{le02}
Let $b>b_{0}$ be a predetermined variable and $g\in C^1(R_+)\cap
C^{2}(R_+\setminus \{b\})$ satisfy the following HJB equations and
the boundary conditions,
\begin{eqnarray}\label{eq17}
&&\max\limits_{U\in[l,1]}\big[\frac{1}{2}\sigma^{2}U^{2}g^{''}(x)+(\mu
U-aU^2-\delta)g^{'}(x)-cg(x)\big ]=0,\\
&& \mbox{ for}\ \  0\leq x\leq
b,\nonumber\\
&& g^{'}(x)=1, \ \mbox{for}\   x\geq b, \nonumber\\
&& g^{''}(x)=0,\ \mbox {for}  \ x> b,\nonumber\\
&& g(0)=0,\nonumber
\end{eqnarray}
then we have the following,
\begin{eqnarray}\label{eq18}
&&\max\mathcal {L}g(x)\leq 0, \ \mbox{ for} \   x \geq
0,\\
&&g^{'}(x)\geq 1, \ \mbox{ for} \   x \geq b,\nonumber \\
&& g(0)=0,\nonumber
\end{eqnarray}
where $g''(b):=g''(b-)$, $\mathcal {L}$ is defined as same as in
Lemma \ref{le01}. Indeed, the function $g(x)$ can be written as
follows,
\begin{eqnarray}\label{eq22}
g(x,b)=\left\{
\begin{array}{l l l}
A\big (e^{\alpha_1x}-e^{\beta_1x}\big ), &0\leq x\leq x_1,\\
\big (Be^{\alpha_2x_2}+Ce^{\beta_2x_2}\big
)\exp\big\{-\int_x^{x_2}\frac{c}{\frac{1}{2}
\mu\eta(y)-\delta}dy\big\},&x_1<x<x_2,\\
Be^{\alpha_2x}+Ce^{\beta_2x}, &x_2\leq x\leq b, \\
x-b+Be^{\alpha_2b}+Ce^{\beta_2b},&x>b,
\end{array}\right.
\end{eqnarray}
where
\begin{eqnarray}\label{eq23}
\alpha_1=\frac{-(\mu l-a l^2-\delta)+\sqrt{(\mu l-a
l^2-\delta)^2+2c\sigma^2 l^2}}{\sigma^2 l^2},
\end{eqnarray}
\begin{eqnarray}\label{eq24}
\beta_1=\frac{-(\mu l-al^2-\delta)-\sqrt{(\mu
l-al^2-\delta)^2+2c\sigma^2 l^2}}{\sigma^2l^2},
\end{eqnarray}
\begin{eqnarray}\label{eq25}
\alpha_2=\frac{-(\mu -a-\delta)+\sqrt{(\mu
-a-\delta)^2+2c\sigma^2}}{\sigma^2},
\end{eqnarray}
\begin{eqnarray}\label{eq26}
\beta_2=\frac{-(\mu -a-\delta)-\sqrt{(\mu
-a-\delta)^2+2c\sigma^2}}{\sigma^2}
\end{eqnarray}
\begin{eqnarray}\label{e260}
A=\frac{Be^{\alpha_2x_2}+Ce^{\beta_2x_2}}
{e^{\alpha_1x_1}-e^{\beta_1x_1}}\exp\big
\{-\int_{x_1}^{x_2}\frac{c}{\frac{1}{2}\mu\eta(y)-\delta}dy\big\},
\end{eqnarray}
\begin{eqnarray}\label{eq30}
B=\frac{[-c+(\frac{1}{2}\mu
1-\delta)\beta_2]e^{\beta_2x_2}}{[-c+(\frac{1}{2}\mu
1-\delta)\beta_2]e^{\beta_2x_2+\alpha_2b}\alpha_2-[-c+(\frac{1}{2}\mu
1-\delta)\alpha_2]e^{\alpha_2x_2+\beta_2b}\beta_2}, \nonumber\\
\end{eqnarray}
\begin{eqnarray}\label{eq31}
C=\frac{[-c+(\frac{1}{2}\mu
1-\delta)\alpha_2]e^{\alpha_2x_2}}{-[-c+(\frac{1}{2}\mu
1-\delta)\beta_2]e^{\beta_2x_2+\alpha_2b}\alpha_2+[-c+(\frac{1}{2}\mu
1-\delta)\alpha_2]e^{\alpha_2x_2+\beta_2b}\beta_2},
\end{eqnarray}
\begin{eqnarray}\label{e261}
x_1=\frac{1}{\alpha_1-\beta_1}\log\big
[\frac{c-\beta_1(\frac{1}{2}\mu l-\delta)}{c-\alpha_1(\frac{1}{2}\mu
l-\delta)}\big]>0,
\end{eqnarray}
\begin{eqnarray}\label{eqx2}
x_2=x_1+\frac{\sigma^2}{2a}\big[\frac{G}{G-H}
\log(\frac{G-l}{G-1})-\frac{H}{G-H}\log(\frac{l-H}{1-H})\big],
\end{eqnarray}
\begin{eqnarray}\label{e262}
G&=&\frac{2c\sigma^2+\mu^2+4a\delta+
\sqrt{(2c\sigma^2+\mu^2+4a\delta)^2-16a\mu^2\delta}}{4a\mu},\\
H&=&\frac{2c\sigma^2+\mu^2+4a\delta-
\sqrt{(2c\sigma^2+\mu^2+4a\delta)^2-16a\mu^2\delta}}{4a\mu},\\
K&=&(G-l)^{G/(G-H)}(l-H)^{-H/(G-H)},
\end{eqnarray}
and $\eta (x)$ is uniquely determined by
\begin{eqnarray}\label{eq28}
[G-\eta(x)]^{G/(G-H)}[\eta(x)-H]^{-H/(G-H)}=K\exp\big
[-\frac{2a}{\sigma^2}(x-x_1)\big ].
\end{eqnarray}
\end{lemma}
\begin{lemma}\label{le03} Let $g(b,x)$ be as the same as in lemma
\ref{le02}. Then $\frac{\partial}{\partial b}g(b,x)\leq 0$  holds
for $b\geq b_0$.
\end{lemma}
\begin{lemma}\label{le04} The ruin probability
 $\mathbf{P}[\tau_b^{b}\leq T]$  is strictly
 increasing w.r.t. $b$ on $ [x_2, b_K)$, where
 $b_K:=\inf\{b: \mathbf{P}[\tau_b^b\leq
T]=0\}$, and $x_2$  is defined by (\ref{eqx2}), $
\tau^b_b:=\tau_{b}^{\pi_b^*}$.
\end{lemma}
\begin{lemma}\label{le05}
Let $\phi^{b}(t,y)\in C^1(0,\infty)\cap C^2(0,b)$ and satisfy the
following partial differential equation and the boundary conditions,
\begin{eqnarray}\label{eq36}
\left\{
\begin{array}{l l l}
\phi_{t}^{b}(t,y)=\frac{1}{2}[U_b^{*}(y)]^2\sigma^2\phi_{yy}^{b}(
t,y)+(\mu U_b^{*}(y)-a[U_b^{*}(y)]^{2}-\delta)\phi_{y}^{b}(t,y),\\
\phi^{b}(0,y)=1,\  \mbox{for}\ \  0<y\leq b, \\
\phi^{b}(t,0)=0,\phi_{y}^{b}(t,y)=0,\ \mbox{for} \ t>0.
\end{array}\right.
\end{eqnarray}
Then $\phi^{b}(T,y)=1-\psi^{b}(T,y)$,
  where  $\psi^{b}(T,y):=
{\bf P}\big \{ \tau_{y}^{b}< T \big\} $, and  $ U^{*}(x)$ is defined
by
\begin{eqnarray}\label{eq21}
U^*_b(x)=\left\{
\begin{array}{l l l} l,&0\leq x\leq x_1,\\
\eta(x)&x_1<x<x_2,\\
1, &x_2\leq x.
\end{array}\right.
\end{eqnarray}
\end{lemma}
\begin{lemma}\label{le06}
Let the function $\phi^{b}(t,x)$ solve  the equation(\ref{eq36}) and
$u(b)\equiv\phi^{b}(T,b)$. Then $u(b)$ is a  continuous  function of
$b$ on $ [b_0, +\infty )$.
\end{lemma}
\vskip 10pt\noindent
 \setcounter{equation}{0}
\section{\bf Proof of Main Result}
\vskip 10pt\noindent In this section we will give the proof of
Theorem \ref{theorem31}. Before this proof we first prove the
following.
\begin{theorem}\label{Th701}
Let $f(x)$, $g(x,b)$ and $U^*_b(x)$ be as the same as in lemma
\ref{le01}, lemma \ref{le02} and lemma
\ref{le05},respectively. \hskip 8pt Then\\
(i) If $b\leq b_0$ we have  $V(x,b)=V(x,b_0)=V(x)=f(x)$, the optimal
strategy associated with $V(x)$ is  $\pi_{b_o}^\ast=\{
 U^*_{b_0}(R^{\pi_{b_o}^\ast}_\cdot), L^{\pi_{b_o}^\ast}_\cdot\}$,
 where the process $\{R^{\pi_{b_o}^\ast}_t, L^{\pi_{b_o}^\ast}_t \}$
 is uniquely determined by the following SDE,
\begin{eqnarray}\label{eq701}
\left\{
\begin{array}{l l l}
dR_t^{\pi_{b_o}^\ast}=(\mu
U^*_{b_0}(R^{\pi_{b_o}^\ast}_t)-a{U^*_{b_0}}^2(R^{\pi_{b_o}^\ast}_t)-\delta
)dt+\sigma U^*_{b_0}(R^{\pi_{b_o}^\ast}_t)d
{W}_{t}-dL_t^{\pi_{b_o}^\ast},\\
R_0^{\pi_{b_o}^\ast}=x,\\
0\leq R^{\pi_{b_o}^\ast}_t\leq  b_0,\\
\int^{\infty}_0 I_{\{t: R^{\pi_{b_o}^\ast}_t
<b_0\}}(t)dL_t^{\pi_{b_o}^\ast}=0.
\end{array}\right.
\end{eqnarray}
(ii) If $b> b_0$ we have  $V(x,b)=g(x)$ and the optimal strategy
$\pi_{b^*}^\ast$ is $
 \{U^\ast_{b^*}(R^{\pi_{b^*}^\ast}_t),L^{\pi_{b^*}^\ast}_t\}$, where
$\{R^{\pi_{b^*}^\ast}_t, L^{\pi_{b^*}^\ast}_t  \} $ is uniquely
determined  by the following SDE
\begin{eqnarray}\label{eq702}
\left\{
\begin{array}{l l l}
dR_t^{\pi_{b^*}^\ast}=(\mu
U^*_{b^*}(R^{\pi_{b^*}^\ast}_t)-a{U^*_{b^*}}^2(R^{\pi_{b^*}^\ast}_t)-\delta
)dt+\sigma U^*_{b^*}(R^{\pi_{b^*}^\ast}_t)d
{W}_{t}-dL_t^{\pi_{b^*}^\ast},\\
R_0^{\pi_{b^*}^\ast}=x,\\
0\leq R^{\pi_{b^*}^\ast}_t\leq  b^*,\\
\int^{\infty}_0 I_{\{t: R^{\pi_{b^*}^\ast}_t
<b^*\}}(t)dL_t^{\pi_{b^*}^\ast}=0.
\end{array}\right.
\end{eqnarray}
\end{theorem}
\begin{proof}
(i) If $b\leq b_0$ then since $\pi^*_{b_0}\in \Pi_{b_0}\subset\Pi_b
$ we have $V(x, b_0)\leq V(x,b)\leq V(x)$. It suffices to show
$V(x)\leq f(x)=V(x, b_0) $. For a admissible strategy $\pi=\{a_\pi,
L^\pi\}$ we assume that $( R^\pi_t, L^\pi_t)$ is  the process
defined by (\ref{eq204}).  Set $\Lambda=\{s:L_{s-}^{\pi}\neq
L_{s}^{\pi}\}$ and let
 $\hat{L}=\sum_{s\in\Lambda, s\leq t}(L_{s}^{\pi}-L_{s-}^{\pi})$
and $\tilde{L}_{t}^{\pi}=L_{t}^{\pi}-\hat{L}_{t}^{\pi}$ denote the
discontinuous part and continuous part of
$L_{s}^{\pi}$,respectively.  Let $\tau^\varepsilon =\inf\{t\geq 0:
R^{\pi}_{t}\leq \varepsilon\}$. Applying It\^{o} formula to
stochastic process $ R^\pi_t $ and $f(x)$, we have
\begin{eqnarray}\label{eq703} e^{-c(t\wedge
\tau^{\varepsilon})}f(R_{t\wedge \tau^{\varepsilon}}^{\pi})&=
&f(x)+\int_{0}^{t\wedge\tau^{\varepsilon}}e^{-cs}\mathcal
{L} f(R_{s}^{\pi})ds\nonumber\\
&+&\int_{0}^{t\wedge\tau^{\varepsilon}}a_\pi\sigma
e^{-cs}f^{'}(R_{s}^{\pi})d\mathcal {W}_{s}\mathcal
-\int_{0}^{t\wedge\tau^{\varepsilon}}e^{-cs}f^{'}(R_{s}^{\pi})dL_{s}^{\pi}\nonumber\\
&+&\sum\limits_{s\in\Lambda ,s\leq t\wedge
\tau^{\varepsilon}}e^{-cs}[f(R_{s}^{\pi})-f(R_{s-}^{\pi})\nonumber\\
&-&f^{'}(R_{s-}^{\pi})(R_{s}^{\pi}-R_{s-}^{\pi})]\nonumber \\
&= &f(x)+\int_{0}^{t\wedge\tau^{\varepsilon}}e^{-cs}\mathcal
{L}  f(R_{s}^{\pi})ds\nonumber\\
&+&\int_{0}^{t\wedge\tau^{\varepsilon}}a_\pi \sigma
e^{-cs}f^{'}(R_{s}^{\pi})d\mathcal {W}_{s}\mathcal
-\int_{0}^{t\wedge\tau^{\varepsilon}}e^{-cs}f^{'}(R_{s}^{\pi})d\tilde{L}_{s}^{\pi}\nonumber\\
&+&\sum\limits_{s\in\Lambda ,s\leq t\wedge
\tau^{\varepsilon}}e^{-cs}[f(R_{s}^{\pi})-f(R_{s-}^{\pi}))],
\end{eqnarray}
where
\begin{eqnarray*}
\mathcal {L}=\frac{1}{2}U^{2}\sigma^{2}\frac{d^{2}}{dx^{2}}+(\mu
U-aU^2-\delta)\frac{d}{dx}-c.
\end{eqnarray*}
By lemma \ref{le01} the second term in the right-hand side of last
equation is nonpositive. Since  $f^{'}(R_{s\wedge
\tau^\varepsilon}^{\pi})\leq f^{'}(\varepsilon)$, the  third term is
a square integrable martingale. Taking expectations on both sides of
Eq.(\ref{eq703}) and then letting $\varepsilon\rightarrow 0$ one has
\begin{eqnarray}\label{eq704}
{\bf E}\big \{e^{-c(t\wedge \tau ^\pi _x)}f(R_{t\wedge \tau ^\pi
_x}^{\pi})\big \}&\leq& f(x)-{\bf E}\big \{\int_{0}^{t\wedge \tau
^\pi _x}e^{-cs}
f^{'}(R_{s}^{\pi})d\tilde{L}_{s}^{\pi}\big \}\nonumber\\
&+&{\bf E}\big\{\sum\limits_{s\in\Lambda ,s\leq t\wedge
\tau ^\pi _x}e^{-cs}[f(R_{s}^{\pi})-f(R_{s-}^{\pi})]\big \}.\nonumber\\
\end{eqnarray}
Since $f^{'}(x)\geq 1$ for $x\geq 0$,
\begin{eqnarray}\label{eq705}
f(R_{s}^{\pi})-f(R_{s-}^{\pi})\leq-(L_{s}^{\pi}-L_{s-}^{\pi}).
\end{eqnarray}
So the inequalities (\ref{eq704}) and (\ref{eq705}) yield
\begin{eqnarray}\label{eq706} {\bf E}\big\{e^{-c(t\wedge \tau ^\pi
_x)}f(R_{t\wedge \tau ^\pi _x}^{\pi})\big\}&+&{\bf
E}\big\{\int_{0}^{ t\wedge \tau ^\pi _x}e^{-cs}dL_{s}^{\pi}\big
\}\leq f(x).
\end{eqnarray}
By the definition of $\tau ^\pi _x $ ,$f(0)=0$ and $f'(x)\geq 1$, it
easily follows  that
\begin{eqnarray}\label{eq707}
\liminf\limits_{t\rightarrow\infty}e^{-c(t\wedge \tau ^\pi
_x)}f(R_{t\wedge
\tau ^\pi _x}^{\pi})&=&e^{-c\tau}f(0)I_{\{\tau ^\pi _x<\infty\}}\nonumber\\
&+& \liminf\limits_{t\rightarrow\infty}e^{-ct}f(R_{t})I_{\{\tau ^\pi
_x =\infty\}}\geq 0.
\end{eqnarray}
So we deduce from the inequalities ( \ref{eq706} ) and (\ref{eq707}
) that
\begin{eqnarray*}
J(x,\pi)={\bf E}[\big\{\int_{0}^{\tau ^\pi _x} e^{-cs}dL_{s}^{\pi}
\}]\leq f(x).
\end{eqnarray*}
Therefore
\begin{eqnarray*}
V(x)\leq f(x).
\end{eqnarray*}
If we choose the control strategy $\pi_{b_o}^\ast=\{
 U^*_{b_0}(R^{\pi_{b_o}^\ast}_\cdot),L^{\pi_{b_o}^\ast}_\cdot\}$ and
 stochastic process $(R^{\pi_{b_o}^\ast}_t, L^{\pi_{b_o}^\ast}_t)$
   as in SDE (\ref{eq701}, the inequalities above
become equalities, so
\begin{eqnarray*}
V(x)\leq f(x)=V(x,b_0).
\end{eqnarray*}
(i) thus follows. \vskip 10pt \noindent (ii) Assume $b\geq b_0$. Let
$( R^\pi_t, L^\pi_t)$ be the process as in (\ref{eq204}) for $\pi\in
\Pi_b $. Then
\begin{eqnarray}\label{eq708}
\left\{
\begin{array}{l l l}
\mathbf{P}\{ R_{s-}^\pi\geq R_{s}^\pi\geq b\}+\mathbf{P}\{b\geq
R_{s-}^\pi\geq R_{s}^\pi\} =1,
\ \forall s\geq 0 ,\\
\mathcal {L}g(R_s^\pi)\leq 0 \  \mbox{ $ s\leq \tau ^\pi _x
=\inf\{t\geq 0: R^\pi_s\leq 0\} $},\\
g'(x)=1 \ \mbox{ $x\geq b$}.
\end{array}\right.
\end{eqnarray}
Replacing $f$  in proof of (i) above  with $g$, then using
(\ref{eq708}) and the same argument as in (i) we can get
$$ V(x,b)\leq g(x).$$
Similarly, letting $\pi_{b}^\ast=\{
 U^*_b(R^{\pi_{b}^\ast}_\cdot),L^{\pi_{b}^\ast}_\cdot\}$
 we derive $V(x,b)=g(x)$. Therefore (ii) follows.
\end{proof}
\vskip 10pt\noindent Now we turn to proof of Theorem
\ref{theorem31}.
 \vskip 10pt\noindent {\bf Proof of Theorem
\ref{theorem31}}.
If $\mathbf{P}[\tau_{b_0}^{\pi^*_{b_0}}\leq
T]\leq\varepsilon$, then the conclusion is obvious because  the
constraints does not work and the proof reduces to  the usual
optimal control problem. \vskip 10pt\noindent
 If $
\mathbf{P}[\tau_{b_0}^{\pi^*_{b_0}}\leq T]>\varepsilon $, then by
lemmas \ref{le04}-\ref{le06} there exists a unique $b^\ast$ solving
equation $\mathbf{P}\{\tau_{b}^{\pi^*_{b}}\leq T\}=\varepsilon $
 and $x_2<b^\ast=\inf\{ b:  b\in \mathfrak{B} \}
> b_0$. By theorem \ref{Th701} we know that  $b^\ast$ meets (\ref{eq302})
 and (\ref{eq303}) because $V(x, b)=g(x,b)$  is a decreasing function of
   $ b(\geq b_0)\in [ x_2, b_K)$ due to
lemma \ref{le03}. So the optimal strategy associated with  the
optimal return function $ V(x,b^*)=g(x, b^*)$ is $
 \{U^\ast_{b^*}(R^{\pi_{b^*}^\ast}_t),L^{\pi_{b^*}^\ast}_t\}$ and
$\{R^{\pi_{b^*}^\ast}_t, L^{\pi_{b^*}^\ast}_t  \} $ is uniquely
determined  by  SDE (\ref{eq304}). Thus we complete the proof. \ \
$\Box$
 \setcounter{equation}{0}
\section{\bf Appendix}
\vskip 10pt\noindent In this section we will give the proofs of
lemmas we concerned with throughout this paper. \vskip 10pt\noindent
{\bf Proof of lemma \ref{le01}}. \ Since  the proof is complete
similar to that of Guo Xin, Liu Jun and Zhou Xunyu\cite{s1}(2004),
we omit it here.
 \vskip 10pt\noindent
{\bf Proof of lemma \ref{le02}}. \ Since the proof is  somewhat
similar to that of He, Hou and Liang \cite{ime02}(2008) and Guo, Liu
and Zhou \cite{s1}(2004), we only give the sketch of the proof as
follows.\\
If the $ \max$ in (\ref{eq17}) is attained in the interior of the
control region, then, by differentiating w.r.t. $U$, we can find the
maximizing function $U^*_b(x)$ can be defined by (\ref{eq21}) above,
and
\begin{eqnarray}\label{eq411}
\eta(x)=\frac{\mu g'(x)}{2ag'(x)-\sigma^2g''(x)} =\frac{2cg(x)}{\mu
g'(x)}+\frac{2\delta}{\mu},0\leq x\leq b.
\end{eqnarray}
Letting $x\rightarrow 0+$ we have  $\eta(x)\rightarrow
2\delta/\mu<l$, whereas taking $x\rightarrow b_0$ and noticing that
$g''(b_0)=0$,  we also have  $\eta(x)\rightarrow \mu/2a>1$. So by
(\ref{eq21}) we find $0\leq x_1<x_2\leq b_0<b$ such that
\begin{eqnarray}\label{e221}
U^*(x)=\left\{
\begin{array}{l l l} l,&0\leq x\leq x_1,\\
\eta(x)&x_1<x<x_2,\\
1, &x_2\leq x\leq b.
\end{array}\right.
\end{eqnarray}
Putting this expression into (\ref{eq17}) we have (\ref{eq22}). Then
by smooth fit principle we can determine parameters $\alpha_1$,
$\beta_1$,$\alpha_2$, $\beta_2$, $x_1$, $x_2$, $A$, $B$, $C$,
$G$,$H$ and $K$ by (\ref{eq23})-(\ref{e262}).
 Now it remains to prove the solution $g$ defined by (\ref{eq22})
  satisfies (\ref{eq18}).
We only need to prove
\begin{eqnarray*}
\max\limits_{U\in[l,1]}[\frac{1}{2}\sigma^{2}U^{2}g^{''}(x)+(\mu
U-aU^2-\delta)g^{'}(x)-cg(x)]=0,\ \mbox{ for}\  x\geq b.
\end{eqnarray*}
For $x\geq b$, we first  prove $g''(b-)\geq 0$. Noticing that
\begin{eqnarray}\label{eq32}
g''(b-)=\frac{\alpha_2^2v(\beta_2)e^{\beta_2x_2+\alpha_2b}-\beta_2^2
v(\alpha_2)e^{\alpha_2x_2+\beta_2b}}
{v(\beta_2)e^{\beta_2x_2+\alpha_2b}\alpha_2-v(\alpha_2)e^{\alpha_2x_2+\beta_2b}\beta_2},\nonumber\\
\end{eqnarray}
where $v(t)\equiv-c+(\frac{1}{2}\mu -\delta)t$. Since
$|\alpha_2|\leq |\beta_2|$,   $v(\beta_2)<v(\alpha_2)<0$, $\alpha_2
>0, \beta_2 <0 $ and $g''(b_0)=0$, the numerator and
denominator of (\ref{eq32}) are strictly negative, so  $g''(b-)\geq
0$ . Then  by (\ref{eq17})
\begin{eqnarray*}\label{eq33}
\max \mathcal {L} \{g(x)\}&=&(\mu -a-\delta)-c(x-b+g(b))\nonumber\\
&\leq & \mu -a-\delta-cg(b)\nonumber\\
&\leq & \frac{1}{2}\sigma^2a^2g''(b-)+\mu -a-\delta-cg(b)\leq 0.
\end{eqnarray*}
Thus we complete the proof.

\vskip 10pt\noindent {\bf Proof of lemma \ref{le03}}.\ If $x\geq b$,
then using (\ref{eq22}) and $g'(x)=1$ one has
\begin{eqnarray}
\frac{\partial}{\partial b}g(b,x)
&=&-1+\alpha_2Be^{\alpha_2b}+
\beta_2Ce^{\beta_2b}+B'e^{\alpha_2b}+C'e^{\beta_2b}\nonumber\\
&=&B'e^{\alpha_2b}+C'e^{\beta_2b}
\end{eqnarray}
where $B'$ and $C'$ denote derivatives w.r.t $b$ of $B$ and
$C$,respectively. By the first three expressions in (\ref{eq22}) the
proof reduces to showing that for $x_2\leq x\leq b$
\begin{eqnarray}\label{e429}
B'e^{\alpha_2x}+C'e^{\beta_2x}\leq 0.
\end{eqnarray}
By (\ref{e262})-(\ref{eq28}) one has
\begin{eqnarray*}\label{eq430}
 B'e^{\alpha_2x}+C'e^{\beta_2x}
&=&\big\{\frac{
[v(\alpha_2)e^{\alpha_2x_2+\beta_2x}-v(\beta_2)e^{\beta_2x_2+\alpha_2x}
]}{
[\alpha_2v(\beta_2)e^{\beta_2x_2+\alpha_2b}-
\beta_2v(\alpha_2)e^{\alpha_2x_2+\beta_2b}]^2}\big\}\\
&&\times\big [
\alpha_2^2v(\beta_2)e^{\beta_2x_2+\alpha_2b}-
\beta_2^2v(\alpha_2)e^{\alpha_2x_2+\beta_2b}
\big ]
\\&=&\frac{K_1(x)}{K_2^2(x_2)}\times K_3(b).\\
\end{eqnarray*}
where
\begin{eqnarray*}
K_1(x)&=&v(\alpha_2)e^{\alpha_2x_2+\beta_2x}-v(\beta_2)e^{\beta_2x_2+\alpha_2x},\\
K_2(x)&=&\alpha_2v(\beta_2)e^{\beta_2x+\alpha_2b}-\beta_2v(\alpha_2)
e^{\alpha_2x+\beta_2b},\\
K_3(x)&=&\alpha_2^2v(\beta_2)e^{\beta_2x_2+
\alpha_2x}-\beta_2^2v(\alpha_2)e^{\alpha_2x_2+\beta_2x}.
\end{eqnarray*}
Since $K_1(x)$ is an increasing function of $x$ and $K_1(x_2)>0$,
$K_1(x)>0$ for $x\in [ x_2, b]$. Noting that $K_3(b)\leq 0$, we know
that (\ref{e429}) is true. \hskip 10pt $\Box$

 \vskip 10pt\noindent
{\bf Proof of lemma \ref{le04}}. \ Because the proof of decreasing
property is complete similar to that of theorem 3.1 in
\cite{ime02}(2008),  we only need to  prove that the probability of
bankruptcy is strictly decreasing on $[x_2 ,b_K]$, that is,
\begin{eqnarray*}
\mathbf{P}[\tau_{b_1}^{b_1}\leq T]-\mathbf{P}[\tau_{b_2}^{b_2}\leq
T]>0
\end{eqnarray*}
 for any $b_2>b_1\geq x_2 $.
By comparison theorem,
$$ \mathbf{P}[\tau_{b_1}^{b_1}\leq T]-\mathbf{P}[\tau_{b_2}^{b_2}\leq T]\geq
\mathbf{P}[\tau_{b_1}^{b_2}\leq T]-\mathbf{P}[\tau_{b_2}^{b_2}\leq
T].
$$
The proof can be reduced to proving that
\begin{eqnarray}\label{59}
\mathbf{P}[\tau_{b_1}^{b_2}\leq T]-\mathbf{P}[\tau_{b_2}^{b_2}\leq
T]>0.
\end{eqnarray}
\vskip 10pt \noindent To prove the inequality (\ref{59}) we define
stochastic processes $R_t^{[1]}$ and $R_t^{[2]}$ by the following
SDEs:
\begin{eqnarray*}
dR_t^{[1]}=[\mu
U_{b_2}^*(R_t^{[1]})-a{U_{b_2}^*}^2(R_t^{[1]})-\delta]dt+U_{b_2}^*(R_t^{[1]})\sigma
d\mathcal {W}_t-dL_t^{b_2},R_0^{[1]}=b_1,
\end{eqnarray*}
\begin{eqnarray*}
dR_t^{[2]}=[\mu
U_{b_2}^*(R_t^{[2]})-a{U_{b_2}^*}^2(R_t^{[2]})-\delta]dt+U_{b_2}^*(R_t^{[2]})\sigma
d\mathcal{W}_t-dL_t^{b_2},R_0^{[2]}=b_2,
\end{eqnarray*}
respectively,  where $ U_b^{*}(\cdot)$ is  as in (\ref{eq21}).
\vskip 10pt \noindent Let $\tau^{b_1}=\inf\limits_{t\geq
0}\{t:R_t^{[2]}=b_1\}$, $A= \{ \tau^{b_1}\leq T\}$ and $B=\big\{\sl
R_{t}^{[2]}$ will go to bankruptcy in a time interval
$[\tau^{b_1},\tau^{b_1}+T]$ and $\tau^{b_1}\leq T \big \}$. Then
$\{\tau_{b_2}^{b_2}\leq T\}\subset B \subset A$. Moreover, by using
strong Markov property of $R_t^{[2]}$, we have
\begin{eqnarray*}
\mathbf{P}[\tau_{b_1}^{b_2}\leq T]=\mathbf{P}[B|A].
\end{eqnarray*}
So
\begin{eqnarray*}
 \mathbf{P}[\tau_{b_1}^{b_2}\leq T]-\mathbf{P}[\tau_{b_2}^{b_2}\leq T]
 &\geq &
\mathbf{P}[\tau_{b_1}^{b_2}\leq T]-\mathbf{P}(B)\nonumber\\
&=& \mathbf{P}[\tau_{b_1}^{b_2}\leq T]-\mathbf{P}(A)\mathbf{P}(B|A)\nonumber\\
&=&\mathbf{P}[\tau_{b_1}^{b_2}\leq T](1-\mathbf{P}(A))\nonumber\\
&=&\mathbf{P}[\tau_{b_1}^{b_2}\leq T]\mathbf{P}(A^c).
\end{eqnarray*}
By theorem \ref{theorem41}, $\mathbf{P}[\tau_{b_1}^{b_2}\leq T]\geq
\mathbf{P}[\tau_{b_1}^{b_1}\leq T]>0$. So we only need to prove
$\mathbf{P}(A^c)>0$.  For doing this we define stochastic processes
$R_t^{[3]}$ and $R_t^{[4]}$ by the following SDEs
\begin{eqnarray*}
\left\{
\begin{array}{l l l}
dR_t^{[3]}=[\mu
U_{b_2}^*(R_t^{[3]})-a{U_{b_2}^*}^2(R_t^{[3]})-\delta]dt+U_{b_2}^*(R_t^{[3]})\sigma
d\mathcal{W}_t-dL_t^{b_2},\\
R_0^{[3]}=\frac{b_1+b_2}{2}
 \end{array}\right.
\end{eqnarray*}
and
\begin{eqnarray*}
\left\{
\begin{array}{l l l}
dR_t^{[4]}=[\mu
U_{b_2}^*(R_t^{[4]})-a{U_{b_2}^*}^2(R_t^{[4]})-\delta]dt+U_{b_2}^*(R_t^{[4]})\sigma
d\mathcal{W}_t,\\
R_0^{[4]}=\frac{b_1+b_2}{2}.
 \end{array}\right.
\end{eqnarray*}
 Setting $D=\{\inf\limits_{0\leq t\leq T}R_t^{[3]}>b_1\}$ and
    $E=\{\inf\limits_{0\leq t\leq T}R_t^{[4]}>b_1,\sup\limits_{0\leq
t\leq T}R_t^{[4]} <b_2\}$, by comparison theorem on SDE, we have
$\mathbf{P}(A^c)\geq \mathbf{P}(D)\geq \mathbf{P}(E)$. Since
$U_{b_2}^*(x)=1$ we have
\begin{eqnarray}\label{510}
R_t^{[4]}= \frac{b_1+b_2}{2}+ [\mu-a-\delta]t+\sigma\mathcal{W}_t \
\mbox{ on $ E$}.
\end{eqnarray}
We deduce from (\ref{510}) and properties of Brownian motion with
drift  (cf. Borodin and Salminen \cite{s110} (2002)) that
\begin{eqnarray*}
\mathbf{P}(E)=\frac{e^{-\mu'^2T/2}}{\sqrt{2\pi
T}}\sum_{k=-\infty}^{\infty}\int_{b_1/\sigma}^{b_2/\sigma}e^{\mu'(z-x)}
[(e^{-(z-x+\frac{2k(b_2-b_1)}{\sigma})^2/2T})\\
-(e^{-(z+x-\frac{2b_1-2k(b_2-b_1)}{\sigma})^2/2T})]dz>0,\nonumber
\end{eqnarray*}
where $\mu'=(\mu-a-\delta)/\sigma$ and $x=\frac{b_1+b_2}{2\sigma}$.
 Thus the
proof follows. $\Box$
 \vskip 10pt \noindent
{\bf Proof of lemma \ref{le05}}.
 Let $\phi(t,y)\equiv \phi^{b}(t,y)$. Since the
stochastic process $(R^{\pi_{b}^\ast,y}_{t\wedge\tau_{y}^{b}},
L_{t\wedge\tau_{y}^{b}}^{\pi_{b}^\ast} )$   is continuous, by
applying the generalized It\^{o} formula to
$(R^{\pi_{b}^\ast,y}_{t\wedge\tau_{y}^{b}},
L_{t\wedge\tau_{y}^{b}}^{\pi_{b}^\ast} )$ and $\phi(t,y)$, we have
for $0<y\leq b$
\begin{eqnarray}\label{eq52}
\phi(T-(t\wedge\tau_{y}^{b}),Y_{t\wedge\tau_{y}^{b}}^{b})
&=&\phi(T,y)\nonumber\\
&+&\int_{0}^{t\wedge\tau_{y}^{b}}(\frac{1}{2}U^{*2}(Y_{s}^{b})
\sigma^{2}\phi_{yy}(T-s,Y_{s}^{b})\nonumber\\
&+&(\mu U^{*}(Y_{s}^{b})-a[U^{*}((Y_{s}^{b}))]^{2}-\delta)
\phi_{y}(T-s,Y_{s}^{b})\nonumber\\
&-&\phi_{t}(T-s, Y_{s}^{b}))ds-\int_{0}^{t\wedge\tau_{y}^{b}}
\phi_{y}(T-s,Y_{s}^{b})dL_{s}^{b}\nonumber\\
&+&\int_{0}^{t\wedge\tau_{y}^{b}}
a(Y_{s}^{b})\sigma\phi_{y}(T-s,Y_{s}^{b})dW_{s}^{b}.
\end{eqnarray}
where $ \tau_{y}^{b}\equiv\tau^{\pi_{b}^\ast}_y=
\inf\{t: R^{\pi_{b}^\ast,y}_t =0\}    $.\\
 Letting $t=T$ and taking mathematical expectation at both
sides of (\ref{eq52}) yield that
\begin{eqnarray*}
\phi(T,y)&=&\mathbf{E}[\phi(T-(T\wedge\tau_{y}^{b}),
R^{\pi_{b}^\ast,y}_{T\wedge\tau_{y}^{b}}
)]\nonumber\\
&=&\mathbf{E}[\phi(0, R^{\pi_{b}^\ast,y}_{T})1_{T<\tau_{y}^{b}}]+
\mathbf{E}[\phi(T-\tau_{y}^{b},0)1_{T\geq\tau_{y}^{b}})]\nonumber\\
&=&\mathbf{E}[1_{T<\tau_{y}^{b}}]=1-\psi(T,y).
\end{eqnarray*}\hskip 10pt
$\Box$
 \vskip 10pt \noindent
{\bf Proof of lemma \ref{le06}}.\ Let
$a(y):=\frac{1}{2}[U_b^{*}(y)]^2\sigma^2$£¬$\mu(y):=\mu
U_b^{*}(y)-a[U_b^{*}(y)]^{2}-\delta$. Then the equation (\ref{eq36})
becomes
\begin{eqnarray}\label{eq53}
\phi_{t}^{b}(t,y)=a(y)\phi_{yy}^{b}( t,y)+\mu(y)\phi_{y}^{b}(t,y).
\end{eqnarray}
By the properties of $U_b^{*}(y)$, we can easily show that $a(y)$
and $\mu(y)$ are continuous in $[0,b]$. So there exists a unique
solution in $C^1(0,\infty)\cap C^2(0,b)$ for (\ref{eq36}). Moreover,
$a^{'}(y)$, $\mu^{'}(y)$ and $a^{''}(y)$ are bounded in
$(0,x_1)$,$(x_1, x_2)$ and $(x_2,b)$. So we only need to prove that
$\phi^{b}(t,x)$ is continuous in $ b $. Let  $y=bz$ and
$\theta^{b}(t,z)=\phi^{b}(t,by)$, the equation (\ref{eq36}) becomes
\begin{eqnarray}\label{eq54}
\left\{
\begin{array}{l l l}
\theta_t^{b}(t,z)=[a(bz)/b^2]\theta_{zz}^{b}(t,z)+[\mu(bz)/b]\theta_z^{b}(t,z),\\
\theta^{b}(0,z)=1, \ \mbox{for}\  0<z\leq 1, \\
\theta^{b}(t,0)=0,\theta_{z}^{b}(t,1)=0,\ \mbox{for} \  t>0.
\end{array}\right.
\end{eqnarray}
So the proof of Lemma \ref{le06} reduces to proving
$\lim\limits_{b_2\rightarrow
b_1}\theta^{b_2}(t,z)=\theta^{b_1}(t,z)$ for fixed $b_1>b_0$.
Setting $w(t,z)=\theta^{b_2}(t,z)-\theta^{b_1}(t,z)$, we  have
\begin{eqnarray}\label{424}
\left\{
\begin{array}{l l l}
w_t(t,z)&=&[a(b_2z)/b_2^2]w_{zz}(t,z)+[\mu(b_2z)/b_2]w_{z}(t,z)\\
&+& \{a(b_2z)/b_2^2-a(b_1z)/b_1^2\}\theta_{zz}^{b_1}(t,z)\\
&+&\{a(b_2z)/b_2^2-a(b_1z)/b_1^2\}\theta_z^{b_1}(t,z),\\
w(0,z)&=&0,\ \mbox{for}\ 0<z\leq 1,\\
w(t,0)&=&0,\ w_x(t,1)=0,\ \mbox{for}\ t>0.
\end{array}\right.
\end{eqnarray}
By multiplying the first equation in (\ref{424}) by $w(t,z)$ and
then integrating  on $ [0,t]\times [0,1]$,
\begin{eqnarray}\label{425}
&&\int_0^{t}\int_0^1 w(s,x)w_t(s,x)dxds
\nonumber\\
&=&\int_0^{t}\int_0^1 \big\{[a(b_2x)/b_2^2]
w(s,x)w_{xx}(s,x)\nonumber\\
&+&[\mu(b_2x)/b_2]w(s,x)w_x(s,x)\nonumber\\
 &+&[a(b_2x)/b_2^2-a(b_1x)/b_1^2]w(s,x)\theta_{xx}^{b_1}(t,x)\nonumber\\
  &+&w(s,x)[\mu(b_2x)/b_2-\mu(b_1x)/b_1]w(s,x)\theta_x^{b_1}(t,x)\big \}dxds\nonumber\\
&\equiv & E_1 +E_2+E_3+E_4.
\end{eqnarray}
We now estimate terms $E_i $, $i=1,\cdots, 4$, as follows. \vskip
5pt \noindent By  definitions of $a(x)$ and $\mu (x) $, there exit
positive constants $D_1 $, $D_2 $ and $ D_3 $ such that
$[\mu(b_2z)/b_2]^2\leq D_1$, $[a(bx)/b^2]'\geq 0$,
$[a(b_2x)/b_2^2]\geq D_2$ and $[a(b_2x)/b_2^2]'\leq D_3$, so by
Young's inequality, we have for any $\lambda_1>0$ and $\lambda_2>0$
\begin{eqnarray}\label{427}
E_1&=&\int_0^{t}\int_0^1[a(b_2x)/b_2^2]w(s,x)w_{xx}(s,x)dxds\nonumber\\
&=&-\int_0^{t}\int_0^1 [a(b_2x)/b_2^2]w_x^2(s,x)dxds\nonumber\\
&&-
\int_0^{t}\int_{0}^{m/b_2}[a(b_2x)/b_2^2]^{'}w_x(s,x)w(s,x)dxds\nonumber\\
&\leq & -D_2\int_0^{t}\int_0^1
w_x^2(s,x)dxds\nonumber\\
&& +D_3\int_0^{t}\int_0^1 [\lambda_1 w_x^2(s,x)+\frac{1}{4\lambda_1}
w^2(s,x)]dxds
\end{eqnarray}
and
\begin{eqnarray}\label{428}
E_2&=&\int_0^{t}\int_0^1[\mu(b_2x)/b_2]w(s,x)w_x(s,x)dxds\nonumber\\
&\leq
&\lambda_2\int_0^{t}\int_0^1 w_x^2(s,x)dxds\nonumber\\
&&+\frac{D_1}{4\lambda_2}\int_0^{t}\int_0^1 w^2(s,x)dxds.
\end{eqnarray}
We decompose $E_3$ as follows.
\begin{eqnarray}\label{eq510}
E_3&=&\int_0^t\int_0^1 \{a(b_2x)/b_2^2-a(b_1x)/b_1^2\}
w(s,x)\theta_{xx}^{b_1}(s,x)dxds\nonumber\\
&=&-\int_0^t\int_0^1 \{a(b_2x)/b_2^2-a(b_1x)/b_1^2\}
w_x(s,x)\theta_{x}^{b_1}(s,x)dxds\nonumber\\
&&-\int_0^t\{\int_0^{x_1/b_2}\{a(b_2x)/b_2^2-a(b_1x)/b_1^2\}'
w(s,x)\theta_{x}^{b_1}(s,x)dx\nonumber\\
&&-\int_{x_1/b_2}^{x_1/b_1}\{a(b_2x)/b_2^2-a(b_1x)/b_1^2\}'
w(s,x)\theta_{x}^{b_1}(s,x)dx\nonumber\\
&&-\int_{x_1/b_1}^{x_2/b_2}\{a(b_2x)/b_2^2-a(b_1x)/b_1^2\}'
w(s,x)\theta_{x}^{b_1}(s,x)dx\nonumber\\
&&-\int_{x_2/b_2}^{x_2/b_1}\{a(b_2x)/b_2^2-a(b_1x)/b_1^2\}'
w(s,x)\theta_{x}^{b_1}(s,x)dx\nonumber\\
&&-\int_{x_2/b_1}^{1}\{a(b_2x)/b_2^2-a(b_1x)/b_1^2\}
w(s,x)'\theta_{x}^{b_1}(s,x)dx\}ds\nonumber\\
&=&-\int_0^t\int_0^1 \{a(b_2x)/b_2^2-a(b_1x)/b_1^2\}
w_x(s,x)\theta_{x}^{b_1}(s,x)dxds\nonumber\\
&&-\int_0^t\{\int_{x_1/b_2}^{x_1/b_1}
\{a(b_2x)/b_2^2-a(b_1x)/b_1^2\}'
w(s,x)\theta_{x}^{b_1}(s,x)dx\nonumber\\
&&-\int_{x_1/b_1}^{x_2/b_2}\{a(b_2x)/b_2^2
-a(b_1x)/b_1^2\}'w(s,x)\theta_{x}^{b_1}(s,x)dx\nonumber\\
&&-\int_{x_2/b_2}^{x_2/b_1}\{a(b_2x)/b_2^2
-a(b_1x)/b_1^2\}'w(s,x)\theta_{x}^{b_1}(s,x)dx\}ds\nonumber\\
&:=& E_{30}+E_{31}+E_{32}+E_{33}.
\end{eqnarray}
It easily follows that
\begin{eqnarray*}
\lim_{b_2\rightarrow b_1}\{|E_{31}| + |E_{33}|\}=0.
\end{eqnarray*}
Since there exists an $L>0$ such that for or all
$x\in(x_1/b_2,x_2/b_1)$
\begin{eqnarray}\label{518}
\left\{
\begin{array}{l l l}
|[a(b_2x)/b_2^2]-[a(b_1x)/b_1^2]|\leq L|b_2-b_1|,\\
|[a(b_2x)/b_2^2]'-[a(b_1x)/b_1^2]'|\leq L|b_2-b_1|,\\
|[\mu(b_2x)/b_2]-[\mu(b_1x)/b_1]|\leq L|b_2-b_1|,
\end{array}\right.
\end{eqnarray}
we have for any $\lambda_3>0$
\begin{eqnarray*}
E_{30}+E_{32}
&=&-\int_0^t\int_0^1\{a(b_2x)/b_2^2-a(b_1x)/b_1^2\}w_x(s,x)
\theta_{x}^{b_1}(s,x)dxds\nonumber\\
&-&\int_0^t\int_{x_1/b_1}^{x_2/b_2}
\{a(b_2x)/b_2^2-a(b_1x)/b_1^2\}'w(s,x)\theta_{x}^{b_1}(s,x)dxds\\
&\leq &\frac{L^2(b_2-b_1)^2}{4\lambda_3}\int_0^t\int_0^1
[\theta_{x}^{b_1}(s,x)]^2dxds\\
&& +\lambda_3\int_0^t\int_0^1 w_x^2(s,x)+w^2(s,x)dxds.
\end{eqnarray*}
By the boundary conditions, we estimate $\int_0^t\int_0^1
[\theta_{x}^{b}(s,x)]^2dxds$  for $b\in [b_1,b_2]$ as follows:
\begin{eqnarray*}
0&=&\int_0^t\int_0^1 \theta_{t}^{b}(s,x)\theta^{b}(s,x)\\
&& -[a(bx)/b^2]\theta_{xx}^{b}(s,x)\theta^{b}(s,x)
-[\mu(bx)/b]\theta_{x}^{b}(s,x)\theta^{b}(s,x)dxds\\
&=&\frac{1}{2}\int_0^1 [\theta^{b}(s,x)]^2dx
+\int_0^t\int_0^1 [a(bx)/b^2] [\theta_{x}^{b}(s,x)]^2dxds\\
&&+\int_0^t\int_0^1
[a(bx)/b^2]'[\theta_{x}^{b}(s,x)][\theta^{b}(s,x)]dxds\\
&&-\int_0^t\int_0^1 [\mu(bx)/b][\theta_{x}^{b}(s,x)][\theta^{b}(s,x)]dxds\\
&\geq & \lambda \int_0^t\int_0^1 [\theta_{x}^{b}(s,x)]^2dxds
-\frac{\lambda}{2}\int_0^t\int_0^1 [\theta_{x}^{b}(s,x)]^2dxds\\
&&-\frac{D_4}{2\lambda}\int_0^t\int_0^1 [\theta^{b}(s,x)]^2dxds\\
&\geq &\frac{\lambda}{2}\int_0^t\int_0^1 [\theta_{x}^{b}(s,x)]^2dxds
-\frac{D_4t}{2\lambda},
\end{eqnarray*}
where $ |[a(bz)/b^2]'-[\mu(bx)/b]|^2< D_4  $  and
$\lambda=\frac{l^2}{8b_1^2}   $, so
\begin{eqnarray*}
\int_0^t\int_0^1 [\theta_{x}^{b}(s,x)]^2dxds\leq
\frac{D_4t}{\lambda^2}.
\end{eqnarray*}
Therefore there exists a function  $B^{b_1}(b_2)$ satisfying
\begin{eqnarray*}
\lim\limits_{b_2\rightarrow b_1}B^{b_1}(b_2)=0,
\end{eqnarray*}
such that   for $ 0\leq t\leq T$
\begin{eqnarray}\label{eq511}
E_3&=&\int_0^t\int_{0}^{1}\{a(b_2x)/b_2^2-a(b_1x)/b_1^2\}
w(s,x)\theta_{xx}^{b_1}(t,x)dxds\nonumber\\
&\leq & B^{b_1}(b_2)+\lambda_3\int_0^t\int_0^1
w_x^2(s,x)+w^2(s,x)dxds.
\end{eqnarray}
Similarly,
\begin{eqnarray}\label{eq512}
E_4&=&\int_0^t\int_0^1 \{\mu(b_2x)/b_2-\mu(b_1x)/b_1\}w(s,x)\theta_x^{b_1}(t,x)dxds\nonumber\\
&\leq &  B_1^{b_1}(b_2)+\int_0^t\int_0^1 w^2(s,x)dxds.
\end{eqnarray}
with
\begin{eqnarray*}
\lim\limits_{b_2\rightarrow b_1}B_1^{b_1}(b_2)=0,
\end{eqnarray*}
Choosing  $\lambda_1$, $\lambda_2$ and $\lambda_3$, by
(\ref{425})-(\ref{428}), (\ref{eq511})-(\ref{eq512}), and
\begin{eqnarray*}\label{eq57}
\int_0^{t}\int_0^1w(s,x)w_t(s,x)dxds&=&\int_0^1
\frac{1}{2}w^2(t,x)dx,
\end{eqnarray*}
there exist constants $C_1$ and $C_2$ such that
\begin{eqnarray*}
\int_0^1w^2(t,x)dx\leq
C_1\int_0^t\int_0^1w^2(s,x)dxds+C_2[B_1^{b_1}(b_2)+B^{b_1}(b_2)].
\end{eqnarray*}
By setting $F(t)=\int_0^t\int_0^1w^2(s,x)dxds$ and using the
Gronwall inequality, we get
\begin{eqnarray*}
F(t)\leq C_2[B_1^{b_1}(b_2)+B^{b_1}(b_2)]\exp\{C_1t\},
\end{eqnarray*}
so
\begin{eqnarray*}
\lim\limits_{b_2\rightarrow
b_1}\int_0^t\int_0^1[\theta^{b_2}(s,x)-\theta^{b_1}(s,x)]^2dxds=0.
\end{eqnarray*}
Thus the proof is complete. \hskip 10pt $\Box$
 \vskip 10pt\noindent
{\bf Acknowledgements.}
 This work
is supported by Projects 10771114 and 11071136 of NSFC, Project
20060003001 of SRFDP, and SRF for ROCS, SEM, and the Korea
Foundation for Advanced Studies. We would like to thank the
institutions for the generous financial support. We also express our
deep thanks to Professor Huaiyu Jian, Wen'an  Yong and  Dr Zhijie
Chen for their valuable comments and suggestions. Special thanks
also go to the participants of the seminar stochastic analysis and
finance at Tsinghua University for their feedbacks and useful
conversations. \vskip 20pt \setcounter{equation}{0}

\end{document}